\documentclass[10pt]{article}
\usepackage{fullpage}
\usepackage{amsmath,amsfonts,amsthm,amssymb,dsfont,bbm,bm}
\usepackage{url}
\usepackage{graphicx}
\usepackage{wrapfig}
\usepackage{caption} 
\usepackage{subcaption}
\usepackage{algorithm}
\usepackage[noend]{algpseudocode}
\usepackage{dsfont}
\usepackage{float}
\usepackage{framed}
\usepackage{comment}
\usepackage{enumerate}
\usepackage[inline]{enumitem}
\usepackage{xcolor}
\usepackage[colorlinks=true, linkcolor=red, urlcolor=blue, citecolor=blue]{hyperref}
\usepackage[makeroom]{cancel}
\usepackage{todonotes}
\usepackage{textcomp}
\usepackage{xcolor}
\usepackage{nicefrac}
\usepackage{braket}
\usepackage{authblk}  
\usepackage{booktabs} 
\usepackage{upgreek}

\usepackage{tikz}
\usetikzlibrary{shapes.geometric, arrows, shapes.symbols, arrows.meta, positioning} 

\DeclareMathOperator*{\argmin}{argmin}
\DeclareMathOperator{\trace}{Tr}

\newtheorem{conjecture}{Conjecture}

\newtheorem{fact}{Fact}
\newtheorem{lemma}{Lemma}

\newtheorem{theorem}{Theorem}
\newtheorem{problem}{Problem}

\theoremstyle{definition}
\newtheorem{definition}{Definition}

\newcommand{\sizeof}[1]{\left\lvert{#1}\right\rvert}

\newcommand{\determinant}[1]{\left\lvert{#1}\right\rvert}
\newcommand{\Tr}[1]{\trace\left({#1}\right)}

\def\FLCuSum{\texttt{FL-CUSUM}}

\title{A Quantum Speedup in Localizing Transmission Loss Change \\
in Optical Networks}

\author[1]{Yufei Zheng\thanks{\texttt{\{yufeizheng, yuzhenchen, towsley\}@cs.umass.edu}}}
\author[1]{Yu-Zhen Janice Chen}
\author[2]{Prithwish Basu\thanks{\texttt{prithwish.basu@rtx.com}}}
\author[1]{Don Towsley}

\affil[1]{\textit{College of Information and Computer Sciences, University of Massachusetts Amherst}}
\affil[2]{\textit{RTX BBN Technologies}}

\date{}  

\begin{document}
\maketitle

\begin{abstract}
The ability to localize transmission loss change to a subset of links in optical networks is crucial for maintaining network reliability, performance and security.
\emph{Quantum probes}, implemented by sending blocks of $n$ coherent-state pulses augmented with continuous-variable (CV) squeezing ($n=1$) or weak temporal-mode entanglement ($n>1$) over a lossy channel to a receiver with homodyne detection capabilities, are known to be more sensitive than their quasi-classical counterparts in detecting a sudden increase in channel loss.
The enhanced sensitivity can be characterized by the increased Kullback-Leibler (KL) divergence of the homodyne output, before and after the loss change occurs.
When combined with the theory of quickest change detection (QCD), the increase in KL divergence translates into a decrease in detection latency.

In this work, we first revisit quantum probes over a channel, and consider the ratio of quantum to classical KL divergences as a metric for quantum speedup.
We generalize previous results on $n=1$ (CV squeezed states) to arbitrary values of $n$, and discuss the monotonicity and asymptotics of this ratio with respect to the parameters of the quantum probe.
Assuming a subset of nodes in an optical network is capable of sending and receiving such probes through intermediate nodes with all-optical switching capabilities, we present a scheme for quickly detecting the links that have suffered a sudden drop in transmissivity. 
Since quantum probes lose their sensitivity with increasing loss in the channel, we first propose a probe construction algorithm that makes the set of links suffering transmission loss change identifiable, while minimizing the longest distance a probe traverses.
We then introduce new cumulative sum (CUSUM) statistics with a stopping rule, which allows us to run the CUSUM algorithm to quickly localize the faulty links using our constructed probes.
Finally, we show that the proposed scheme achieves a quantum speedup in decreasing the localization delay.
\end{abstract}

\section{Introduction}

\emph{If we augment part of a classical network with quantum capabilities, what can we do better than in the purely classical setting?}
In this paper, we show that detecting transmission loss change in an optical network provides one answer to this question.

Optical networks are highly prevalent in modern communication systems.
Major cloud providers rely on optical networks to connect data centers and ensure reliable access to cloud services~\cite{korsbäck2023growing,poutievski2022jupiter,velush2023boosting}. Additionally, optical mesh networks, e.g., Reconfigurable Optical Add-Drop Multiplexer (ROADM) networks have seen significant density of deployment by metropolitan area Internet Service Providers and telecom providers.
However, unattended optical fiber networks are often subjected to physical damages to fibers, or unauthorized fiber tapping~\cite{medard1997security,skorin2016physical,Karlsson:23}.
These lead to increases in transmission loss, which often translates into increased bit error rates, which may in turn impact the quality of service.
Perhaps more strikingly, by introducing only a small transmission loss, an eavesdropper can extract information on the data being transferred~\cite{medard1997security,Karlsson:23}.
Therefore, the ability to detect transmission loss in an optical network is crucial for maintaining network reliability, performance and security.

Classical and quasi-classical methods such as bit error rate measurement, optical spectrum analysis (OSA), and mean optical power monitoring exist for detecting changes in transmission loss in an optical \emph{channel}, but are generally not fast enough to detect a change in transmission loss that lasts for only seconds~\cite{skorin2016physical,gong2020secure}.
Optical time-domain reflectometry (OTDR) is used for fault detection and localization, but it is not sensitive enough for the order of loss incurred by a stealthy eavesdropper~\cite{chan2010optical,shim2012correlation,gong2020secure}.
The Quantum Alarm technique~\cite{gong2020secure} achieves better empirical sensitivity, by using the parameter estimation step of continuous-variable quantum key distribution (CV-QKD) protocols to act as an alarm for loss detection.
The quasi-classical benchmark in~\cite{guha2025quantum} is a better protocol due to the data processing inequality---while the parameters used in Quantum Alarm~\cite{gong2020secure} are the output of homodyne detection after post-processing---the quasi-classical benchmark in~\cite{guha2025quantum} works directly with the probabilistic distributions from the homodyne output, using the theory of \emph{quickest change detection (QCD)}.

On the quantum side, \cite{guha2025quantum} recently shows that, a vanishingly small quantum energy in the form of \emph{squeezed states}, when added to the bright classical pulses, dramatically increases the sensitivity in detecting transmission loss change in an optical channel. 
The quantum protocols in~\cite{guha2025quantum} are simple to implement and effective, with the simplest protocol involving sending displaced squeezed states as probes from a squeezing augmented transmitter over a lossy channel, and processing the output at a homodyne detection receiver.
In the advanced protocol, coherent states augmented by a continuous-variable entanglement produced by a temporal-mode Green Machine~\cite{cui2023superadditive} are sent as probes, and the outputs are similarly captured by a homodyne detector.
Feeding the homodyne outputs into a \emph{cumulative sum (CUSUM) algorithm}~\cite{page1954continuous}, one can quickly catch the change, with a detection delay that is inversely proportional to the \emph{Kullback-Leibler (KL) divergence} from the pre-change to the post-change distributions. 
Compared to the quasi-classical benchmark, the quantum probes induce a larger KL divergence, which leads to a lower detection delay.

In this paper, we first revisit quantum probes over a channel, generalize and extend existing theoretical results on squeezing-augmented probes to the general entanglement-augmented probes.
Then, building on the quantum probes for a channel~\cite{guha2025quantum}, we propose a method for detecting and localizing transmission loss change in a \emph{network}.
Though the technologies for sending and receiving these probes are readily available, it is only realistic to assume that a subset of the network nodes are equipped with such quantum capabilities.
We refer to the quantum-capable nodes as \emph{monitors}.
We assume the rest of the nodes, the non-monitors, are passive optical switches that can pass the quantum probes through.

When expanding the goal from change detection in a channel to fault localization in a network, we must decide how to route the probes given a network topology and a set of monitors.
If we think of whether a link suffers a sudden transmission drop as a Boolean state of the link, the framework of \emph{Boolean network tomography}~\cite{duffield2003simple} can help infer every link's state with limited visibility through probes between monitors, and therefore localize the transmission drop.
While the probes do not output Boolean values, their ability to \emph{distinguish} whether a change occurs on the links they traverse is what makes the Boolean framework relevant.
Previous works have studied where to place monitors in a network~\cite{ma2015optimal}, conditions that allow network failures to be identified~\cite{ma2014node}, and how to construct probes to localize faults~\cite{ahuja2011srlg,stanic2010active}.
Recognizing that we often do not get to choose where quantum capabilities are placed, in this work, we focus on the probe construction problem.
The fact that probes are quantum brings in new objectives in designing the construction algorithms.
Due to factors such as attenuation, noise and decoherence, quantum probes are inherently sensitive to the the distance they travel.
While existing heuristics aim to minimize the total cost~\cite{ahuja2011srlg,stanic2010active} or the total length of probes in the context of classical networks, we introduce a new probe construction algorithm that optimizes the distance the longest probes traverse.

Assuming all observations from the probes are collected at a central site, we develop a \emph{network-level CUSUM} algorithm to quickly identify the problematic links.
Despite the CUSUM algorithm being a typical choice for detecting change~\cite{guha2025quantum,zou2019quickest,zou2018quickest}, much less is known on how to apply it in the setting of fault localization in a network.
For simplicity, we restrict our attention to the case where at most one link in a network suffers from a transmission drop.
To the best of our knowledge, the only previous attempt is a CUSUM-based approach proposed to identify outages in a power network~\cite{rovatsos2017statistical}, which outputs candidates of problematic links, and relies on further simulations to narrow down the search.
Using the probes generated by our construction algorithm, we are able to adapt the CUSUM algorithm to directly pinpoint the faulty link. 

Our contributions in this paper are twofold:
\begin{itemize}
\item We generalize and extend existing results on squeezing-augmented probes to the general entanglement-augmented probes, and characterize the quantum speedup in reducing change detection delay over a channel (\S~\ref{sec:probes}).
\item We present a scheme with quantum speedup in localizing transmission drop in a network, which incorporates a probe construction algorithm (\S~\ref{sec:tomography}), and a CUSUM-based fault localization algorithm (\S~\ref{sec:qcd}).
\end{itemize}
We evaluate the performance of our network-level CUSUM algorithm on both a line network, and an all-optical version of the 3-tier fat-tree network topology~\cite{al2008scalable} (\S~\ref{sec:simulation}). 
Finally, we discuss related work in \S~\ref{sec:related}, and conclude the paper in \S~\ref{sec:conclusion}.
\section{Quantum and Classical Probes over a Channel} \label{sec:probes}

In this section, we consider quantum and classical probes that are sent through a lossy channel with transmissivity $\eta_0 \in (0,1)$.
Suppose at an unknown time $\nu$, the transmissivity of the channel drops from $\eta_0 = \eta \in (0,1)$ to $\eta_1 = \eta_0 \eta_{d}$, $\eta_{d} \in (0,1)$.
We first review existing results from~\cite{guha2025quantum} in \S~\ref{sec:probes:review}, and then proceed to derive the quantum speedup in reducing detection delay over a channel in \S~\ref{sec:probes:results}.

\subsection{Preliminaries on probes} \label{sec:probes:review}

We begin with a brief introduction to the quantum and classical probes used in this paper, followed by a short description of how the probes are used to detect a transmissivity drop in a channel.

\paragraph{Pre-change and post-change distributions.}
As in~\cite{guha2025quantum}, the classical (or quasi-classical) benchmark we consider is a stream of $T$ coherent-state pulses $\ket{\alpha}^{\otimes T}$, with $\alpha \in \mathbb{R}$ and mean photon number $\alpha^2 = N + N_a$.
Here $0 < N_a << N$ is introduced to fairly compare the classical case with its quantum-augmented counterparts.
At the receiving monitor, $\ket{\alpha}^{\otimes T}$ generates $T$ i.i.d. Gaussian-distributed random variables $X_1, X_2, \dots, X_T$, where 
\begin{equation} \label{eqn:classical_Gaussian}
X_t \sim \mathcal{N}^{(j)}(\sqrt{\eta_j} \alpha, \, \frac{1}{4}), 
\end{equation}
and $j = 0$ for all $0 \leq t < \nu$ (pre-change), $j = 1$ for all $\nu \leq t \leq T$ (post-change).

Next, we consider the entanglement-augmented probes.
Specifically, $T$ blocks of $n \in \mathbb{Z}_+$ coherent-state pulses $\ket{\alpha}^{\otimes n}$ are augmented by a continuous variable entangled state generated by splitting a squeezed-vaccum state $\ket{0;s}$, $s \in \mathbb{R}_{+}$, where 
\begin{align}
& N = \alpha^2, \label{eqn:N_alpha} \\
& \sinh^2(s) = n N_a \label{eqn:N_a_s}.
\end{align}
In this case, the receiving monitor outputs $T$ i.i.d. $n$-dimensional Gaussian random vectors $\bm{X}_1, \bm{X}_2, \dots, \bm{X}_T$, where 
\begin{equation} \label{eqn:quantum_Gaussian}
\bm{X}_t \sim \mathcal{N}_n^{(j)}(\bm{\mu}^{(j)}, \, \bm{\Sigma}^{(j)}),
\end{equation}
and $j = 0$ pre-change, $j = 1$ post-change.
Denote $\bm{I}$ as the $n \times n$ identity matrix, $\bm{J}$ as the $n \times n$ matrix of ones, and $\bm{u}$ as the $n$-dim column vector of ones. 
The mean vectors $\bm{\mu}^{(j)}$ and the covariance matrices $\bm{\Sigma}^{(j)}$ are given by
\begin{align}
& \bm{\mu}^{(j)} = \sqrt{\eta_j} \alpha \cdot \bm{u}, \nonumber \\
& \bm{\Sigma}^{(j)} = \frac{1}{4} \bm{I} - \frac{\eta_j (1-e^{-2s})}{4n} \cdot \bm{J}. \label{eqn:cov}
\end{align}

In particular, when $n=1$, we have a stream of $T$ displaced squeezed states $\ket{\alpha;r}^{\otimes T}$, with $\alpha^2 = N$ and $\sinh^2(r) = N_a$.
The receiving monitor outputs $T$ i.i.d. Gaussian-distributed random variables $X_1, X_2, \dots, X_T$, where 
\begin{equation} \label{eqn:sqeezed_Gaussian}
X_t \sim \mathcal{N}_1^{(j)}(\sqrt{\eta_j} \alpha, \, \sigma_j^2),
\end{equation}
and $\sigma_j^2 = \frac{1}{4} (\eta_j e^{-2r} + 1-\eta_j)$, and $j = 0$ pre-change, $j = 1$ post-change.
While the physical processes of generating the probes for $n=1$ and $n \geq 2$ are different~\cite{guha2025quantum}, mathematically, there is in fact no need to distinguish these two cases.
We only make the distinction when discussing previous results in relation to ours in \S~\ref{sec:probes:results}.
Whenever necessary, we refer to $n=1$ as the squeezing augmentation case, and $n \geq 2$ the general quantum augmentation case. 

\paragraph{Quickest change detection.}
Let $X_1, X_2, \dots, X_T$ be the observations from any type of probes introduced above.
Suppose a transmission drop occurs at an unknown timestep $0< \nu \leq T$, and $f_0(x),f_1(x)$ are the pre-change and post-change distributions respectively (e.g., (\ref{eqn:classical_Gaussian}) for classical probes, (\ref{eqn:quantum_Gaussian}) for quantum probes).
To detect the change point $\nu$, the cumulative sum (CUSUM) algorithm keeps a counter $\text{CUSUM}[t]$ that can be updated recursively using
\[
\text{CUSUM}[t] = \max\left\{0, \text{CUSUM}[t-1] + \log{\frac{f_1(X_t)}{f_0(X_t)}}\right\}, \forall t \in [T],
\]
and $\text{CUSUM}[0]=0$.
Let $\gamma$ be the minimum time to false alarm that one is willing to tolerate, and denote $\tau$ as the minimum $t$ such that $\text{CUSUM}[t] \geq \log{\gamma}$.
Then the CUSUM algorithm outputs $\tau$ as the estimate for the change point $\nu$.
Informally, as $\gamma \to \infty$, the CUSUM algorithm achieves a worst-case detection delay~\cite{L1971} of 
\begin{equation} \label{exp:WADD_channel}
\frac{\log{\gamma}}{D(f_1 || f_0)} (1+o(1)),
\end{equation}
where $D(f_1 || f_0) = \int_{-\infty}^{\infty} f_1(x) \log{\left(\frac{f_1(x)}{f_0(x)} \right)} \,dx$ is the \emph{Kullback-Leibler (KL) divergence} from $f_0$ to $f_1$.

\paragraph{Kullback-Leibler (KL) divergence.}
Given two $n$-dimensional Gaussian distributions $\mathcal{N}_n^{(j)}(\bm{\mu}^{(j)}, \, \bm{\Sigma}^{(j)})$, $j=0,1$, with non-singular covariance matrices $\bm{\Sigma}^{(j)}$ and $\bm{\Sigma}^{(j)}$, the KL divergence from $\mathcal{N}_n^{(0)}$ to $\mathcal{N}_n^{(1)}$ is 
\begin{equation} \label{eqn:KL_quantum_n}
D\left(\mathcal{N}_n^{(1)}||\mathcal{N}_n^{(0)}\right) 
= \frac{1}{2} 
\left( \Tr{\left(\bm{\Sigma}^{(0)}\right)^{-1} \bm{\Sigma}^{(1)}} 
+ \left(\bm{\mu}^{(0)} - \bm{\mu}^{(1)}\right)^{T} \left(\bm{\Sigma}^{(0)}\right)^{-1} \left(\bm{\mu}^{(0)} - \bm{\mu}^{(1)}\right) 
-n 
+ \ln{\frac{\determinant{\bm{\Sigma}^{(0)}}}{\determinant{\bm{\Sigma}^{(1)}}}}\right),
\end{equation}
where $\determinant{\bm{\Sigma}^{(j)}}$ denotes the determinant of $\bm{\Sigma}_j$, $j = 0,1$.
By (\ref{eqn:classical_Gaussian}) and (\ref{eqn:KL_quantum_n}) in the case of $n=1$, the KL divergence for each coherent state pulse is
\begin{equation} \label{eqn:KL_classical}
D(\mathcal{N}^{(1)}||\mathcal{N}^{(0)})
= 2 (N + N_a) (\sqrt{\eta_0} - \sqrt{\eta_1})^2
= 2 (N + N_a) \eta (1 - \sqrt{\eta_{d}})^2.
\end{equation}
To get the per-pulse KL divergence for entanglement-augmented coherent states, we take the average over $n$ pulses in that block.
That is, we use $\frac{1}{n} D\left(\mathcal{N}_n^{(1)}||\mathcal{N}_n^{(0)}\right) $ as a fair comparison to the classical case (\ref{eqn:KL_classical}), with $D\left(\mathcal{N}_n^{(1)}||\mathcal{N}_n^{(0)}\right)$ as shown in (\ref{eqn:KL_quantum_n}).

\subsection{The quantum speedup in reducing detection delay} \label{sec:probes:results}

Since the worst-case detection delay is inversely proportional to the KL divergence~\eqref{exp:WADD_channel}, to quantify how quantum probes fare against the classical benchmark in speeding up the change point detection, we consider the ratio of quantum to classical KL divergences. 
To start with, we derive the explicit form of this ratio as a function of $\eta \in (0,1)$, $\eta_{d} \in (0,1)$, $N \in \mathbb{Z}_{+}$ and $N_a \in \mathbb{Z}_{+}$, parametrized by $n \in \mathbb{Z}_{+}$.
\begin{lemma} \label{lem:q_n}
Let $q_n = q_n(\eta, \eta_{d}, N, N_a) = \frac{D\left(\mathcal{N}_n^{(1)}||\mathcal{N}_n^{(0)}\right)}{n \cdot D(\mathcal{N}^{(1)}||\mathcal{N}^{(0)})}$ be the ratio of quantum to classical KL divergences, 
and denote $c_n = \frac{2\sqrt{nN_a}}{\sqrt{nN_a+1}+\sqrt{nN_a}}$, 
\begin{equation} \label{eqn:q_n}
q_n
= \frac{1}{4(N+N_a)n} \left( 
\frac{c_n}{1- c_n \eta} \cdot \frac{1-\eta_{d}}{(1-\sqrt{\eta_{d}})^2}
-\frac{1}{\eta} \cdot \frac{\ln{\frac{1 - c_n\eta\eta_{d}}{1 - c_n\eta }}}{(1-\sqrt{\eta_{d}})^2}
+ \frac{4Nn}{1-c_n \eta}
\right).
\end{equation}
\end{lemma}

\begin{proof}
First note that from (\ref{eqn:cov}), covariance matrix $\bm{\Sigma}^{(j)}$ has uniform diagonal entries $d_j = \frac{1}{4}\left(1- \frac{\eta_j (1-e^{-2s})}{4n} \right)$ and uniform off-diagonal entries $r_j = d_j - \frac{1}{4}$.
Such a structure allows us to explicitly write out both $\determinant{\bm{\Sigma}^{(j)}}$ and $\left(\bm{\Sigma}^{(j)}\right)^{-1}$, where
\[
\determinant{\bm{\Sigma}^{(j)}} = 
\begin{vmatrix}
d_j+(n-1)r_j & d_j+(n-1)r_j & \cdots & d_j+(n-1)r_j \\
r_j          & d_j          & \cdots & r_j \\
\vdots       & \vdots           & \ddots & \vdots \\
r_j          & r_j          & \cdots & d_j
\end{vmatrix}
= (\frac{1}{4} + n r_j) \cdot \frac{1}{4^{n-1}}.
\]
Hence the last term in (\ref{eqn:KL_quantum_n}) is
\begin{equation} \label{eqn:temp_ln_Sigmas}
\ln{\frac{\determinant{\bm{\Sigma}^{(0)}}}{\determinant{\bm{\Sigma}^{(1)}}}}
= - \ln{\frac{1-\eta\eta_{d}(1-e^{-2s})}{1-\eta(1-e^{-2s})}}.
\end{equation}
By (\ref{eqn:N_a_s}), $e^{-2s} = 1-2\sqrt{nN_a}(\sqrt{nN_a+1}-\sqrt{nN_a})$, and
\begin{equation} \label{eqn:s_c_n}
1-e^{-2s} 
= \frac{2\sqrt{nN_a}}{\sqrt{nN_a+1}+\sqrt{nN_a}}
\triangleq c_n.
\end{equation}
Substitution of (\ref{eqn:s_c_n}) into (\ref{eqn:temp_ln_Sigmas}) yields
\begin{equation} \label{eqn:ln_Sigmas}
\ln{\frac{\determinant{\bm{\Sigma}^{(0)}}}{\determinant{\bm{\Sigma}^{(1)}}}}
= - \ln{\frac{1 - c_n\eta\eta_{d}}{1 - c_n\eta }}.
\end{equation}

To derive $\left(\bm{\Sigma}^{(0)}\right)^{-1}$, observe that $\bm{\Sigma}^{(j)}$ (\ref{eqn:cov}) is the sum of a diagonal matrix and a rank-$1$ matrix, that is, $\bm{\Sigma}^{(j)} = \frac{1}{4} \bm{I} + \bm{u} (r_j\bm{u})^T$.
By the Sherman-Morrison formula,
\[
\left(\bm{\Sigma}^{(0)}\right)^{-1}
= \left(\frac{1}{4} \bm{I}\right)^{-1} - \frac{\left(\frac{1}{4} \bm{I}\right)^{-1} \cdot \bm{u} (r_0\bm{u})^T \cdot \left(\frac{1}{4} \bm{I}\right)^{-1}}{1 + (r_0\bm{u})^T \left(\frac{1}{4} \bm{I}\right)^{-1} \bm{u}}
= 4\bm{I} + \frac{4c_n \eta}{n(1-c_n \eta)}\bm{J}.
\]
Then we are ready to write out the first two term in (\ref{eqn:KL_quantum_n}),
\begin{equation} \label{eqn:trace}
\Tr{\left(\bm{\Sigma}^{(0)}\right)^{-1} \bm{\Sigma}^{(1)}}
= n + \frac{c_n \eta}{1- c_n \eta}(1-\eta_{d}),
\end{equation}
\begin{equation} \label{eqn:2nd_term}
\left(\bm{\mu}^{(0)} - \bm{\mu}^{(1)}\right)^{T} \left(\bm{\Sigma}^{(0)}\right)^{-1} \left(\bm{\mu}^{(0)} - \bm{\mu}^{(1)}\right)
= (\sqrt{\eta_0}-\sqrt{\eta_1})^2 
\cdot \alpha^2 
\cdot \bm{u}^T \left(\bm{\Sigma}^{(0)}\right)^{-1} \bm{u}
= \frac{4Nn\eta (1-\sqrt{\eta_{d}})^2}{1-c_n \eta},
\end{equation}
where the last equality in (\ref{eqn:2nd_term}) uses (\ref{eqn:N_alpha}).
Substituting (\ref{eqn:trace}), (\ref{eqn:2nd_term}) and (\ref{eqn:ln_Sigmas}) into (\ref{eqn:KL_quantum_n}), we arrive at the final form of $\frac{1}{n}D\left(\mathcal{N}_n^{(1)}||\mathcal{N}_n^{(0)}\right)$,
\begin{equation} \label{eqn:KL_quantum_n_final}
\frac{1}{n}D\left(\mathcal{N}_n^{(1)}||\mathcal{N}_n^{(0)}\right)
= \frac{1}{2n} \left( 
\frac{c_n \eta}{1- c_n \eta}(1-\eta_{d})
- \ln{\frac{1 - c_n\eta\eta_{d}}{1 - c_n\eta }}
+ \frac{4Nn\eta}{1-c_n \eta} (1-\sqrt{\eta_{d}})^2
\right).
\end{equation}
Therefore, the form of $q_n$ in (\ref{eqn:q_n}) follows from dividing the per-pulse quantum KL divergence (\ref{eqn:KL_quantum_n_final}) by the classical KL divergence (\ref{eqn:KL_classical}).
\end{proof}

Using the explicit form of the per-pulse quantum KL divergence (\ref{eqn:KL_quantum_n_final}), we are able to rigorously derive several results that have been previously inferred through numerical evaluations in~\cite{guha2025quantum}.
It has been shown that the KL divergence sees a sharp increase when augmenting a tiny amount of squeezing.
That is, when $n=1$, $\lim_{N_a \to 0} \frac{\partial}{\partial N_a} \frac{1}{n}D\left(\mathcal{N}_n^{(1)}||\mathcal{N}_n^{(0)}\right) = \lim_{N_a \to 0} \frac{\partial}{\partial N_a} D\left(\mathcal{N}_1^{(1)}||\mathcal{N}_1^{(0)}\right) = \infty$.
It has also been observed without proof that the same holds for general $n$.
By taking the partial derivative of (\ref{eqn:KL_quantum_n_final}), it is straightforward to verify that indeed,
\[
\lim_{N_a \to 0} \frac{\partial}{\partial N_a} \frac{1}{n}D\left(\mathcal{N}_n^{(1)}||\mathcal{N}_n^{(0)}\right)
= 2N\sqrt{n} \eta^2 (1-\sqrt{\eta_{d}})^2 \lim_{N_a \to 0} \frac{1}{\sqrt{N_a}}
=\infty.
\]
Additionally, numerical evaluations (e.g. Figure 2 in~\cite{guha2025quantum}, with $n=2^i$, $i \in [8]$) suggest that for fixed $N$, $N_a$, $\eta$ and $\eta_{d}$, the quantum KL divergence $\frac{1}{n}D\left(\mathcal{N}_n^{(1)}||\mathcal{N}_n^{(0)}\right)$ converges to $\frac{2N\eta(1-\sqrt{\eta_{d}})^2}{1-\eta}$ as $n$ becomes large.
This directly follows from taking the limit of (\ref{eqn:KL_quantum_n_final}) as $n$ goes to infinity.

With Lemma~\ref{lem:q_n} in hand, we are well-equipped to understand, given a  set of parameters, whether quantum probes outperform the classical ones, and if so, by how much.
Moreover, we can explore how the quantum speedup factor changes with respect to each variable, and what is the asymptotic behavior of such speedup.

\begin{lemma} \label{lem:q_n_monotonicity}
$q_n = q_n(\eta, \eta_{d}, N, N_a)$ exhibits the following properties:
\begin{enumerate}[label={(\arabic*)}]
\item \label{lem:q_n_monotonicity:eta}
For fixed values of $n$, $\eta_{d}$, $N$ and $N_a$, $q_n$ increases in $\eta$, and
\begin{equation} \label{eqn:lim_eta_to_0}
\lim_{\eta \to 0} q_n
= \frac{N}{N+N_a};
\end{equation}

\item \label{lem:q_n_monotonicity:eta_d}
For fixed values of $n$, $\eta$, $N$ and $N_a$, $q_n$ increases in $\eta_{d}$, and 
\begin{equation} \label{eqn:lim_eta_tap_to_0}
\lim_{\eta_{d} \to 1} q_n
= \frac{1}{2n(N+N_a)(1-c_n\eta)}\left( \frac{c_n^2 \eta}{1-c_n \eta} + 2nN \right);
\end{equation}

\item \label{lem:q_n_monotonicity:N}
For fixed values of $n$, $\eta$, $\eta_{d}$ and $N_a$, $q_n$ increases in $N$ when $N > \frac{c_n^2 t (1+\sqrt{\eta_{d}})^2}{4nN_a (1-c_n \eta \eta_{d})}$, and
\begin{equation} \label{eqn:lim_qn_N}
\lim_{N \to \infty} q_n = \frac{1}{1-c_n\eta};
\end{equation}

\item \label{lem:q_n_monotonicity:n}
For fixed values of $\eta$, $\eta_{d}$, $N$ and $N_a$ satisfying $N\eta>b_dN_{a}\left(1-\eta\right)$ with $b_d = \frac{1+\sqrt{\eta_{d}}}{1-\sqrt{\eta_{d}}}$, $q_n$ increases in $n$ when $n \geq n_0$, where
\begin{equation} \label{eqn:n_0}
n_0 =
\begin{cases}
\max\left( \frac{b_d\left(3-4\eta\right)}{4\left(N\eta-b_dN_{a}\left(1-\eta\right)\right)},  \frac{1}{4N_a(1-\eta)}\right)
& \text{if }  0 < \eta < \frac{1}{2}, \\
\max\left( \frac{b_d}{4\left(N\eta-b_dN_{a}\left(1-\eta\right)\right)},
\frac{1}{4N_a(1-\eta)}\right)
& \text{if } \frac{1}{2} \leq \eta < 1,
\end{cases}
\end{equation}
and 
\begin{equation} \label{eqn:lim_n_q_n}
\lim_{n \to \infty} q_n = \frac{N}{(N+N_a)(1-\eta)};
\end{equation}

\item
For fixed values of $n$, $\eta$, $\eta_{d}$ and $N$:
    \begin{enumerate}[label={(5.\arabic*)}]
    \item \label{lem:q_n_monotonicity:Na:advantage} 
    For any $\varepsilon > 0$, denote $c_{n,\varepsilon} = \frac{2\sqrt{n\varepsilon}}{\sqrt{n\varepsilon+1}+\sqrt{n\varepsilon}}$, if $\varepsilon \leq N_a \leq \frac{c_{n,\varepsilon}N\eta}{1-c_{n,\varepsilon}\eta}$, we have $q_n >1$;
    \item 
    $\lim_{N_a \to 0} \frac{\partial}{\partial N_a} q_n = \sqrt{n} \eta \lim_{N_a \to 0} \frac{1}{\sqrt{N_a}} = \infty;$
    \item \label{lem:q_n_monotonicity:Na:dec} 
    If $8Nn\left(1-\eta\right)>\eta b_d\left(1-\eta_{d}\right)$ with $b_d$ defined as above, $q_n$ decreases in $N_a$ when $N_a \geq N_{a,0}$, where
    \begin{equation} \label{ineq:ub_Na}
    \resizebox{.83\textwidth}{!}{$
    \begin{aligned}
    N_{a,0} = 
    &  \frac{8Nn\left(3\eta-1\right)+b_d\left(4\eta-1\right)}{4n\left(8Nn\left(1-\eta\right)-\eta b_d\left(1-\eta_{d}\right)\right)} \\
    & + \frac{\sqrt{8Nn\left(b_d+4Nn\eta\right)\left(8Nn\left(1-\eta\right)-\eta b_d\left(1-\eta_{d}\right)\right)+\left(8Nn\left(3\eta-1\right)+b_d\left(4\eta-1\right)+4\right)^{2}}}{4n\left(8Nn\left(1-\eta\right)-\eta b_d\left(1-\eta_{d}\right)\right)}.
    \end{aligned}
    $}
    \end{equation}
    \end{enumerate}
\end{enumerate}
\end{lemma}

For all claims on the monotonicity, we bound the corresponding partial derivatives. We refer the reader to Appendix A in the full version of our paper~\cite{zheng2025quantum} for the proof of Lemma~\ref{lem:q_n_monotonicity}. 

Lemma~\ref{lem:q_n_monotonicity}\ref{lem:q_n_monotonicity:eta}\ref{lem:q_n_monotonicity:N} not only corroborates an observation that quantum augmentation provides the most benefit in the large-$N$ and high-$\eta$ regime~\cite{guha2025quantum} when $n=1$, but also shows that the same holds in the general quantum augmentation case (arbitrary $n$).
Additionally, \cite{guha2025quantum} provides an approximate threshold $\frac{N\eta}{1-\eta}$ on $N_a$, above which squeezing augmentation does not outperform the classical benchmark in the case of $n=1$.
Lemma~\ref{lem:q_n_monotonicity}\ref{lem:q_n_monotonicity:Na:advantage} gives a range of $N_a$ where quantum speedup exists, for any $n$.
We note that, for typical values of $N$, $n$, $\eta$ and $\eta_d$ (e.g., see \S~\ref{sec:simulation} and~\cite{guha2025quantum}), $q_n$ is likely to first increase in $N_a$, and then decrease.
Lacking ways to directly find the $N_a$ that maximizes $q_n$, we settle for the result in Lemma~\ref{lem:q_n_monotonicity}\ref{lem:q_n_monotonicity:Na:dec}.
Nonetheless, $N_{a,0}$ may be of practical interest, as it narrows down the range of $N_a$ in which the optimal quantum speedup factor appears.
Too see that, given typical values such as $N=100$, $n=1$, $\eta=0.9$ and $\eta_d=0.8$, Lemma~\ref{lem:q_n_monotonicity}\ref{lem:q_n_monotonicity:Na:advantage} gives that the quantum speedup exists when $0<N_a<897.5$ by taking $\epsilon$ ranging from arbitrarily close to $0$ to $890$, and \cite{guha2025quantum} shows that such speedup disappears when $N_a$ is above roughly $900$; however, both bounds on $N_a$ are rendered irrelevant under the assumption that $N_a << N=100$.
In contrast, Lemma~\ref{lem:q_n_monotonicity}\ref{lem:q_n_monotonicity:Na:dec} indicates that quantum augmentation brings the most value when $N_a < N_{a,0} =20.9$.
This is evidence that assuming $N_a << N$ (\S~\ref{sec:probes:review}) makes sense.

Beyond generalizing and strengthening previous results, we can draw a few interesting conclusions from Lemma~\ref{lem:q_n_monotonicity}.
While smaller transmissivity drops (larger $\eta_d$) are generally harder to detect, the quantum speedup factor increases as the drop gets lighter (Lemma~\ref{lem:q_n_monotonicity}\ref{lem:q_n_monotonicity:eta_d}).
Moreover, with proper choices of parameters, the quantum speedup exists, and the limits in Lemma~\ref{lem:q_n_monotonicity}\ref{lem:q_n_monotonicity:eta_d}\ref{lem:q_n_monotonicity:N}\ref{lem:q_n_monotonicity:n} provide the respective upper bounds on how much quantum augmentation could help.

In spite of the generality of Lemma~\ref{lem:q_n_monotonicity}, we acknowledge that the lower bounds of $n$ (in Lemma~\ref{lem:q_n_monotonicity}\ref{lem:q_n_monotonicity:n}) and $N_a$ (in Lemma~\ref{lem:q_n_monotonicity}\ref{lem:q_n_monotonicity:Na:advantage}) are likely consequences of our proof technique, rather than being necessary.
We leave further improvements for future work.
\section{Probe Construction for a Network} \label{sec:tomography}

A network is a weighted graph $G = (V,E)$, with a set of \emph{monitors} $M \subseteq V$ that can \emph{send} and \emph{receive} probes.
If we think of each edge $e$ in $G$ as a channel with transmissivity $\eta_e$, the edge weight $w(e)$ of $e$ is set to $-\log{\eta_e}$.
Throughout this paper, we use the term \emph{link} in the context of a network, but refer to a link as an edge when we discuss the network in graph-theoretic terms.
In this section, we first formulate the problem in the language of Boolean network tomography (\S~\ref{sec:tomography:formulation}), and later introduce an algorithm for constructing the probes (\S~\ref{sec:tomography:algo}).

\subsection{The probe construction problem} \label{sec:tomography:formulation}

Throughout this section, we refer to transmissivity drops as \emph{faults}, and assume only a subset of $E$ in $G$ may be \emph{faulty}.
A \emph{probe} $P$ is a \emph{walk}\footnote{A \emph{walk} with endpoints $v_0$ and $v_k$ is a sequence $v_0, e_1, v_1, e_2, ..., e_k, v_k$, where $v_i \in V$ for all $0\leq i\leq k$, and $e_j \in E$ for all $1\leq i\leq k$.} with endpoints in the set of monitors $M$.\footnote{This definition corresponds to the Arbitrarily Controllable Routing (ACR) scenario in the classical tomography literature (\S~1.4 in~\cite{he2021network}).}
For example, cycles are allowed, and each edge can be traversed in both directions.
When the context is clear, we also think of $P$ as the set of edges the probe passes through, and the \emph{length} of $P$ is $l(P) = \sum_{e\in P} w(e) = -\log{(\prod_{e \in P} \eta_e)}$.
A probe $P$ is faulty, if an only if there exists an edge $e \in P$ that is faulty.
As we shall see in \S~\ref{sec:qcd}, the state of a probe is in fact an observation drawn from the pre- and post-change distributions described in \S~\ref{sec:probes}.
And the Boolean nature of the probes described in this section corresponds to whether the pre- and post-change distributions differ.

We call $F \subseteq E$ a \emph{fault set}, if all edges in $F$ are simultaneously faulty.
For a set of probes $\mathcal{P}$, the fault set $F$ is manifested through $\mathcal{P}_F = \{P \in \mathcal{P} \mid P \cap F \neq \emptyset \}$.
Two fault sets $F_1$ and $F_2$ are \emph{distinguishable} if and only if $\mathcal{P}_{F_1} \neq \mathcal{P}_{F_2}$.
A graph $G$ is \emph{identifiable} with respect to a set of all possible faults $\mathcal{F}$ given a set of monitors $M$, if for any pair of distinct fault sets $F_1, F_2 \in \mathcal{F}$, $F_1$ and $F_2$ are distinguishable.
In practice, we often do not get to select the network topology (the graph $G$), which nodes are quantum capable and can serve as monitors (monitor set $M$), and what part of a network may be faulty (the set of all possible faults $\mathcal{F}$).
For simplicity, we assume $G$ is identifiable w.r.t. the given $M$ and $\mathcal{F}$.\footnote{Relaxing this assumption is straightforward, the proposed algorithm can be easily adapted to identify a maximal subset of $\mathcal{F}$.}
And our goal is to construct a set of probes $\mathcal{P}$, such that every pair of distinct fault sets in $\mathcal{F}$ are distinguishable.

The problem of constructing $\mathcal{P}$ is very well-studied, both in the context of classical network tomography~\cite{cho2014localizing,stanic2010active}, and combinatorial group testing~\cite{cheraghchi2012graph,harvey2007non,spang2018unconstraining}, with the former literature focusing on finding heuristics to minimize the total length $\sum_{P \in \mathcal{P}} l(P)$ of the probes, and the latter putting an emphasis on the optimality of $\sizeof{\mathcal{P}}$.
The challenge posed by our problem is that, quantum probes are inherently sensitive to the distance they travel.
As a probe traverses through multiple edges in the graph, the transmissivities of edges (channels) multiply, and from \S~\ref{sec:probes:results}, it is always desirable to avoid long probes and low transmissivity whenever possible.
Therefore, we are interested in finding the set of probes $\mathcal{P}^{*}$ that minimizes the length of the longest probe:
\begin{problem}[Probe construction] \label{prob}
Construct $\mathcal{P}^{*} = \arg\min_{\mathcal{P}} \max_{P \in \mathcal{P}} l(P)$, where $\mathcal{P}$ is a set of probes such that every pair of distinct fault sets in $\mathcal{F}$ are distinguishable given $G$ and $M$.
\end{problem}

\subsection{Optimizing the longest probes}
\label{sec:tomography:algo}

\paragraph{The meta algorithm.}
Existing algorithms for constructing $\mathcal{P}$ commonly follow a meta algorithm shown as Algorithm~\ref{algo:meta}~\cite[Algorithm 25]{he2021network}.
This meta algorithm iteratively goes through each pair of distinct fault sets $F_1$ and $F_2$, and invokes a subroutine (Step~\ref{algo:meta:FindProbe}) to add a probe that distinguishes $F_1$ and $F_2$.
Here the tag $t_F$ (Step~\ref{algo:meta:t_F}) keeps track of $\mathcal{P}_F$ as the algorithm proceeds, and can be thought of as the syndrome of the fault set $F$. 
If $t_{F_1} \neq t_{F_2}$, $F_1$ and $F_2$ are distinguishable, as $\mathcal{P}_{F_1} \neq \mathcal{P}_{F_2}$.
Also observe that once $t_{F_1} \neq t_{F_2}$ starting from some iteration $i$, they remain unequal till the end, as $t_{F_1}$ and $t_{F_2}$ differ in the $i$ less significant bits.
\begin{algorithm}
\caption{A meta algorithm for probe construction~\cite{he2021network}} 
\label{algo:meta}
\textbf{Input}: A graph $G=(V,E)$, a set of all possible faults $\mathcal{F}$, and a set of monitors $M$; \\
\textbf{Output}: A set of probes $\mathcal{P}$.
\begin{algorithmic}[1]
\State $\mathcal{P} \gets \emptyset$
\State $t_F \gets 0$, for all $F \in \mathcal{F}$ \label{algo:meta:t_F}
\For{$F_1, F_2 \in \mathcal{F}$ such that $F_1 \neq F_2$} \label{algo:meta:loop_start}
    \If{$t_{F_1} = t_{F_2}$}
        \State $P = \Call{FindProbe}{G, M, F_1, F_2}$ \label{algo:meta:FindProbe}
        \For{$F \in \mathcal{F}$ such that $F$ is traversed by $p$}
            \State $t_{F} \gets t_{F} + 2^{\sizeof{\mathcal{P}}}$ 
        \EndFor
        \State $\mathcal{P} \gets \mathcal{P} \cup P$ \label{algo:meta:loop_end}
    \EndIf
\EndFor
\State \Return $\mathcal{P}$
\end{algorithmic}
\end{algorithm}

We also follow Algorithm~\ref{algo:meta}, and the problem of constructing $\mathcal{P}$ boils down to designing an appropriate subroutine \textsc{FindProbe} (Step~\ref{algo:meta:FindProbe}) that finds a probe $P$ given distinct fault sets $F_1$ and $F_2$.
Such a probe $P$ must contain some edge in $F_1 \setminus F_2$ and no edge in $F_2$, or vice versa. 
Interestingly, as we shall see in \S~\ref{sec:tomography:algo}, by making the locally optimal choice in each iteration of the meta algorithm, it is enough to guarantee that the final output $\mathcal{P}$ is a solution to Problem~\ref{prob}.

\paragraph{The subroutine.}
Algorithm~\ref{algo:FindProbe} describes the \textsc{FindProbe} soubroutine.
Given distinct fault sets $F_1$ and $F_2$, recall that the probe $P$ we are looking for must satisfy 
\begin{enumerate*}[label=(\arabic*)]
\item $P \cap (F_1 \setminus F_2) \neq \emptyset$ and $P \cap F_2 = \emptyset$, or \label{case1}
\item $P \cap (F_2 \setminus F_1) \neq \emptyset$ and $P \cap F_1 = \emptyset$.
\end{enumerate*}
Therefore, we find a $P$ with the minimal weight in each case, and finally the better of the two candidates.
Since the two cases are completely symmetric, we focus on case~\ref{case1}.
The only ingredient we need is the information on the shortest paths from monitors to each vertex in the graph.
We then opt for a version of the Floyd-Warshall algorithm (e.g., \S~25.2 in~\cite{cormen2022introduction}), which not only returns the all-pair shortest distances, but also allows the construction of the shortest paths using a predecessor array.
\begin{algorithm}
\caption{$\textsc{FindProbe}(G, M, F_1, F_2)$}
\label{algo:FindProbe}
\begin{algorithmic}[1]
\State Run Floyd-Warshall with path construction on $G_1 = (V, E \setminus F_2)$, obtain distance matrix $D_1$ and predecessor array $T_1$.
\Comment{$D_1(i,j)$ is the shortest distance from $i$ to $j$, and $T_1(i,j)$ gives the penultimate vertex on the path from $i$ to $j$, both in $G_1$.}
\label{algo:FindProbe:FW}
\State $P_1 \gets $ \Call{FindOptProbe}{$G_1, M, F_1 \setminus F_2, D_1, T_1$} \label{algo:FindProbe:P1}
\State Run Floyd–Warshall with path construction on $G_2 = (V, E \setminus F_1)$, obtain $D_2$ and $T_2$.
\State $P_2 \gets $ \Call{FindOptProbe}{$G_2, M, F_2 \setminus F_1, D_2, T_2$}
\State \Return The shorter of $P_1$ and $P_2$
\end{algorithmic}
\end{algorithm}

In yet another subroutine (Algorithm~\ref{algo:FindOptProbe}), we iteratively go through every edge $(u,v)$ in $F_1\setminus F_2$, construct the shortest probe that goes through $(u,v)$ in a graph that does not contain edges from $F_2$.
The notation $m_u \rightsquigarrow u \rightsquigarrow v \rightsquigarrow m_v$ in Step~\ref{algo:FindOptProbe:P} refers to the concatenation of the shortest path from $m_u$ to $u$, edge $(u,v)$, and the shortest path from $v$ to $m_v$.
It follows that $P_1$ (Step~\ref{algo:FindProbe:P1} in Algorithm~\ref{algo:FindProbe}) is the shortest probe satisfying conditions of case~\ref{case1}.

\begin{algorithm}
\caption{$\textsc{FindOptProbe}(G, M, F, D, T)$}
\label{algo:FindOptProbe}
\begin{algorithmic}[1]
\State Initialize $P_0 \gets \emptyset$, and set $w(P_0) = \infty$
\For{$(u,v) \in F$}
    \State $m_u \gets \argmin_{m \in M} D(m,u)$
    \State $m_v \gets \argmin_{m \in M} D(v,m)$
    \State Construct $P$ to be $m_u \rightsquigarrow u \rightsquigarrow v \rightsquigarrow m_v$ using $T$ \label{algo:FindOptProbe:P}
    \If{$w(P) < w(P_0)$}
        \State $P_0 \gets P$
    \EndIf
\EndFor
\State \Return $P_0$
\end{algorithmic}
\end{algorithm}

\paragraph{Correctness.}
Next we show that Algorithm~\ref{algo:meta} with subroutines~\ref{algo:FindProbe} and~\ref{algo:FindOptProbe} correctly outputs a solution to Problem~\ref{prob}.
We say a probe $P$ that distinguishes $F_1$ and $F_2$ is \emph{optimal}, if the length $l(P)$ is minimal among all probes that distinguish $F_1$ and $F_2$.
First we have the correctness of the \textsc{FindProbe} subroutine.

\begin{fact} \label{fact:FindProbe}
Given $G$, $M$, and distinguishable $F_1$, $F_2$, \textsc{FindProbe}$(G, M, F_1, F_2)$ (Algorithm~\ref{algo:FindProbe}) outputs an optimal $P$ that distinguishes $F_1$ and $F_2$.
\end{fact}
\begin{proof}
The fact follows from the correctness of the Floyd-Warshall algorithm and \textsc{FindOptProbe} (Algorithm~\ref{algo:FindOptProbe}) as explained above.
\end{proof}

However, not all pairs of fault sets invoke the \textsc{FindProbe} subroutine.
Changing the order in which the pairs appear in Step~\ref{algo:meta:loop_start} of the meta algorithm~\ref{algo:meta} changes the output $\mathcal{P}$.
Nonetheless, a probe with the maximal length in $\mathcal{P}$ must be generated by some pair $F_1$ and $F_2$, and such a probe is all we need to argue about.
To formalize the argument, it is convenient to introduce the following definition.

\begin{definition} \label{def:max_hard}
Distinct fault sets $F_1$ and $F_2$ in $\mathcal{F}$ are \emph{maximally hard} to distinguish given $G$ and $M$, if for any other pair of distinct fault sets $F_1'$ and $F_2'$, any $P'$ that optimally distinguishes $F_1'$ and $F_2'$ has length $l(P') \leq l(P)$, where $P$ is any probe that distinguishes $F_1$ and $F_2$.
\end{definition}
Following Definition~\ref{def:max_hard}, there may exist multiple pairs of fault sets in $\mathcal{F}$ that are maximally hard to distinguish. 
To show the correctness of the algorithm (Lemma~\ref{lem:correctness}), we only need to focus on the first pair.

\begin{lemma} \label{lem:correctness}
Let $G$ be identifiable with respect to $\mathcal{F}$ and monitors $M$.
Running Algorithm~\ref{algo:meta} on $G$, $M$ and $\mathcal{F}$ with subroutines in Algorithms~\ref{algo:FindProbe} and~\ref{algo:FindOptProbe} gives $\mathcal{P}^{*} = \arg\min_{\mathcal{P}} \max_{P \in \mathcal{P}} l(P)$, where $\mathcal{P}$ is a set of probes such that every pair of distinct fault sets in $\mathcal{F}$ are distinguishable.
\end{lemma}
\begin{proof}
Suppose Algorithm~\ref{algo:meta} first encounters a maximally-hard-to-distinguish pair of fault sets $F_1$ and $F_2$ at iteration $i$ (Step~\ref{algo:meta:loop_start}).
Denote $\mathcal{P}_i$ as the set of probes $\mathcal{P}$ at the end of iteration $i$ (Step~\ref{algo:meta:loop_end}). 
By Definition~\ref{def:max_hard} and Fact~\ref{fact:FindProbe}, $F_1$ and $F_2$ cannot be distinguished by any $P \in \mathcal{P}_i$.
Therefore, the algorithm proceeds to Step~\ref{algo:meta:FindProbe} and finds $P_0 = \textsc{FindProbe}(G, M, F_1, F_2)$, and the fact that $F_1$ and $F_2$ are maximally hard to distinguish guarantees $\max_{P \in \mathcal{P}^*} l(P) = l(P_0)$.
Then for any set of probes $\mathcal{P}$ with $\max_{P \in \mathcal{P}} l(P) < l(P_0)$, no $P \in \mathcal{P}$ distinguishes $F_1$ and $F_2$.
\end{proof}

\paragraph{Time complexity.}
We start with steps relevant to the Floyd-Warshall algorithm.
In Step~\ref{algo:FindProbe:FW} of Algorithm~\ref{algo:FindProbe}, running Floyd-Warshall on a graph with $\sizeof{V}$ vertices takes time $O(\sizeof{V}^3)$.
Step~\ref{algo:FindOptProbe:P} of Algorithm~\ref{algo:FindOptProbe} involves using the output of the Floyd-Warshall to construct two shortest paths of length at most $O(\sizeof{V})$ each, resulting in time $O(\sizeof{V})$.
Suppose $\mathcal{F}$ contains fault sets of size at most $k$, Step~\ref{algo:FindProbe:P1} of Algorithm~\ref{algo:FindProbe} takes time $O\left(\sizeof{F_1 \setminus F_2}\cdot (\sizeof{M}+\sizeof{V})\right)
= O(k\sizeof{V})$.
Algorithm~\ref{algo:FindProbe} has time complexity $O(\sizeof{V}^3+k\sizeof{V})$.
Finally, the probe construction algorithm has time complexity $O(\sizeof{\mathcal{F}}^2 (\sizeof{\mathcal{F}} + \sizeof{V}^3 + k\sizeof{V}))$.

In particular, when $k=1$, we may have only one faulty edge, and $\sizeof{\mathcal{F}} = O(\sizeof{E})$.
In this case, the time complexity of Algorithm~\ref{algo:FindProbe} is dominated by that of the Floyd-Warshall algorithm, and the probe construction algorithm takes time $O(\sizeof{E}^2 (\sizeof{V}^3 + \sizeof{E}))$.
\section{Quickest Change Detection and Localization in a Network} \label{sec:qcd}

Given a set of probes $\mathcal{P}$ with fixed parameters ($N$, $N_a$ and $n$) that can distinguish any pair of faults in $\mathcal{F}$, in this section, we consider how to use the observations from the probes to quickly localize where the fault occurs.
Throughout this section, we assume at most one edge is faulty, that is, $k=1$ and $\mathcal{F} = E \cup \emptyset$.
Relaxing this assumption warrants a separate discussion.
We further assume the transmissivity drop $\eta_d$ on the faulty edge is known.
For an unknown $\eta_d$, it suffices to combine the techniques introduced in this section with that in~\cite{hare2025keeping}.
Next, we formalize the quickest change detection and localization (QCDL) problem in \S~\ref{sec:qcd:problem} and introduce a CUSUM-based algorithm for solving QCDL in \S~\ref{sec:qcd:algo}.

\subsection{Quickest change detection and localization problem}
\label{sec:qcd:problem}

We begin by introducing notations similar to those in \S~\ref{sec:probes:review}.
Let $\mathcal{P}$ be a set of probes such that every pair of distinct faults in $\mathcal{F} = E \cup \emptyset$ is distinguishable.
$\mathcal{P}$ can be constructed by the algorithms detailed in \S~\ref{sec:tomography:algo}, or $\mathcal{P}$ can be any set of probes that guarantees identifiability.
Each probe $P \in \mathcal{P}$ gives a sequence of independent observations $\{X_{P,t}\}_{t\geq 0}$, and observations $\{X_{P',t}\}_{t\geq 0}$ from any $P'$ such that $P' \cap P = \emptyset$ are clearly independent of $\{X_{P,t}\}_{t\geq 0}$.
When $P' \cap P = e$ for some edge $e$, $\{X_{P',t}\}_{t\geq 0}$ and $\{X_{P,t}\}_{t\geq 0}$ are independent conditioned on the state of $e$.
We assume all observations are collected at a central site without delay.

Suppose an edge $e^* \in \mathcal{F}$ suffers a transmission drop at an unknown deterministic time $\nu \geq 1$.
Let $f_{0,P}$ and $f_{1,P}^{(e^*)}$ denote the pre- and post-change distributions for probe $P$, respectively.
We note that, if $e^* \notin P$, that is, the probe $P$ does not traverse $e^*$, we have $f_{1,P}^{(e^*)} = f_{0,P}$.
Then for all $P \in \mathcal{P}$, $X_{P,t} \sim f_{0,P}$ for $t < \nu$, and 
$X_{P,t} \sim f_{1,P}^{(e^*)}$ for $t \geq \nu$.
Let ${\bm X}_t$ be a random vector representing all observations at time $t$, ${\bm X}_t = (X_{P,t})_{P\in \mathcal{P}}$.
Then ${\bm X}_t$ is drawn from pre-change joint distribution ${\bm f}_0 = \prod_{P\in \mathcal{P}} f_{0,P}$ for $t < \nu$, and from the post-change joint distribution ${\bm f}_{1}^{(e^*)} = \prod_{P\in \mathcal{P}} f^{(e^*)}_{1,P}$ for $t \geq \nu$.
Denote $\mathbb{E}_0$ and $\mathbb{E}_\nu^{(e^*)}$ as the expectation corresponding to the joint pre-change distribution ${\bm f}_0$ and post-change distribution ${\bm f}_{1}^{(e^*)}$, respectively.
When the change point $\nu = \infty$, we have $\mathbb{E}_\nu^{(e^*)} = \mathbb{E}_0$, which we also denote by $\mathbb{E}_\infty$.

If there is a faulty edge $e^*$, a fault detection and localization algorithm should declare that an edge $\lambda$ is faulty at time $\tau$, on seeing all observations $\{{\bm X}_t\}_{0 \leq t \leq \tau}$ up to time $\tau$.
Ideally, we want the localization to be accurate, that is, $\lambda = e^*$, and the detection delay $\tau - \nu$ to be positive and small.
$\tau$ is considered a \emph{false alarm}, if $\tau < \nu$.
In the QCD literature, $\tau$ is commonly referred to as the \emph{stopping time}.
Due to the trade-off between false alarm rates and detection speed, the performance guarantees of QCD algorithms are often established in the asymptotic regime, as the worst-case average run length goes to infinity.

Along the same lines, we define our false alarm and localization requirement.
Let $\mathcal{C}_{\gamma}$ be the subfamily of stopping time and localization pairs for which the localization is accurate, and the WARL to false alarm is at least $\gamma$, 
\begin{align}\label{eq:WARL-Localization}
    \mathcal{C}_\gamma = \{(\tau, \lambda) \mid \text{ as } \gamma \rightarrow \infty: \mathbb{E}_\infty[\tau] \geq \gamma \text{ and } \lambda \rightarrow e^* \text{ when } \nu < \infty \}.
\end{align}
By Pollak's criterion~\cite{pollak1985optimal}, for any $(\tau, \lambda) \in \mathcal{C}_\gamma$, the worst-case detection delay is defined as $\text{WADD}(\tau) = \sup\limits_{\nu \geq 1} \mathbb{E}_{\nu}^{(e^*)}[\tau - \nu | \tau \geq \nu]$.
First we state the following universal lower bound on $\text{WADD}(\tau)$ for any $(\tau, \lambda)$ that fulfills the false alarm and localization requirement in~\eqref{eq:WARL-Localization}. 
\begin{theorem}[Detection delay lower bound]\label{thm:universal}
Let $e^* \in E$ be the faulty edge.
As $\gamma \rightarrow \infty$, we have
\begin{align*}
    \inf\limits_{(\tau, \lambda) \in \mathcal{C}_\gamma} \text{WADD}(\tau) \geq \frac{\log \gamma}{\sum_{P \in \mathcal{P}_{e^*}}D\left(f_{1,P}^{(e^*)}\big|\big|f_{0,P}\right)}(1+o(1)).
\end{align*}
\end{theorem}
\noindent
The proof of Theorem~\ref{thm:universal} is deferred to Appendix~\ref{sec:proofs-qcd:lb}.
Next in \S~\ref{sec:qcd:algo}, we are interested in solving the following problem, which we refer to as the quickest change detection and localization (QCDL) problem:
\begin{problem}[QCDL]
Find a stopping rule $(\tau^*, \lambda^*) \in \mathcal{C}_\gamma$ that minimizes $\text{WADD}(\tau)$ as $\gamma \to \infty$. 
\end{problem}

\subsection{Detection and localization algorithm} \label{sec:qcd:algo}

Recall that we assume the parameters $N$, $N_a$, $n$ and $\eta_d$ are all given, and the pre-change transmissivity of each probe $P \in \mathcal{P}$ can be computed from the length $l(P)$ (\S~\ref{sec:tomography:formulation}).
This means that, for a set of probes $\mathcal{P}$, we can effectively compute $f_{0,P}$ and $f_{1,P}^{(e)}$ for every edge $e \in E$ and every probe $P \in \mathcal{P}$.
Consequently, we have the joint distributions ${\bm f}_{0}$ and ${\bm f}_{1}^{(e)}$ for every $e \in E$.
However, without knowing which edge $e^*$ is faulty, or whether a faulty $e^*$ exists, all we can say about the joint post-change distribution ${\bm f}_1^{(e^*)}$ is that ${\bm f}_1^{(e^*)}$ belongs to a known collection $\{{\bm f}_1^{(e)}\}_{e\in E} \cup \{{\bm f}_0\}$, without being able to pinpoint which specific distribution in the collection is ${\bm f}_1^{(e^*)}$.

\paragraph{Direct application of generalized likelihood ratio (GLR) test.}
In this case, a natural approach is to apply a GLR-based algorithm~\cite{L1998}.
At each time $t$, the algorithm maintains the following statistic
\begin{equation} \label{eq:GLR}
\max\limits_{\substack{1\leq j \leq t, \\ e \in E}} \sum_{i=j}^t \log \left(\frac{{\bm f}_1^{(e)}({\bm X}_i)}{{\bm f}_0({\bm X}_i)}\right), 
\end{equation}
and outputs $\tau^{\text{GLR}}$ and $\lambda^{\text{GLR}}$ as estimates for the change point $\nu$ and the faulty edge $e^*$ respectively, where 
\begin{align*}
\tau^{\text{GLR}} & = \inf\left\{t \geq 1: \max\limits_{\substack{1\leq j \leq t, \\ e \in E}} \sum_{i=j}^t \log \left(\frac{{\bm f}_1^{(e)}({\bm X}_i)}{{\bm f}_0({\bm X}_i)}\right) \geq h\right\}, \\
\lambda^{\text{GLR}} & = \arg\max\limits_{\substack{1\leq j \leq t, \\ e \in E}} \sum_{i=j}^{\tau^{\text{GLR}}} \log \left(\frac{{\bm f}_1^{(e)}({\bm X}_i)}{{\bm f}_0({\bm X}_i)}\right).
\end{align*}
Unfortunately, maintaining (\ref{eq:GLR}) directly is inefficient, both in terms of memory consumption and computational complexity, as (\ref{eq:GLR}) has to be recomputed from the ground up using ${\bm X}_1, \dots, {\bm X}_t$ at each timestep $t$.

\paragraph{Fault-Localization CUSUM (\FLCuSum).} 
We now introduce our \FLCuSum\, algorithm (Algorithm~\ref{algo:fl-CUSUM}).
While (\ref{eq:GLR}) cannot be computed efficiently, part of (\ref{eq:GLR}), $\max\limits_{1\leq j \leq t} \sum_{i=j}^t \log \left(\frac{{\bm f}_1^{(e)}({\bm X}_i)}{{\bm f}_0({\bm X}_i)}\right)$, can in fact be updated recursively.
Therefore, instead of maximizing over all possible post-change distributions, we maintain in parallel a separate test statistic $\text{CUSUM}_e$ for each edge $e \in E$,
\begin{equation}
    \text{CUSUM}_e[t] = \max\limits_{1\leq j \leq t}\sum_{i=j}^t \log \left(\frac{{\bm f}_1^{(e)}({\bm X}_i)}{{\bm f}_0({\bm X}_i)}\right),
\end{equation}
which we update recursively using
\begin{align*}
    \text{CUSUM}_e[t] = \max\left\{0,  \text{CUSUM}_e[t-1] +  \log \left(\frac{{\bm f}_1^{(e)}({\bm X}_t)}{{\bm f}_0({\bm X}_t)}\right) \right\}, 
\end{align*}
and $\text{CUSUM}_e[0] = 0$.
\FLCuSum\, halts as soon as an edge-specific statistics $\text{CUSUM}_{\lambda}$, for some $\lambda \in E$, crosses a pre-specified threshold $h$ at time $\tau$.
Then \FLCuSum\, declares that a transmission drop has occurred on edge $\lambda$ at time $\tau$, where we have
\begin{align}
    &\tau = \min\limits_{e \in E} \left\{ \inf\{ t \geq 1 : \text{CUSUM}_e[t] \geq h \} \right\}, \label{eq:flcusum-stopping-time}\\
    &\lambda = \arg\max\limits_{e \in E} \text{CUSUM}_e[\tau]. \nonumber
\end{align}

\begin{algorithm}[t]
\caption{Fault-Localization CUSUM (\FLCuSum) Algorithm} 
\label{algo:fl-CUSUM}
\textbf{Input}: Observations $\{{\bm X}_t\}_{t \geq 0}$, threshold $h$, joint distributions ${\bm f}_0$ and ${\bm f}_1^{(e)}$ for all $e \in E$;\\
\textbf{Output}: Change point estimate $\tau$, and faulty edge estimate $\lambda$.
\begin{algorithmic}[1]
\State $t \gets 0$, $\text{CUSUM}_{e}[t] \gets 0, \forall e \in E$
\While{1}
    \State $t \gets t + 1$
    \For{$e \in E$}
        \State $\text{CUSUM}_e[t] \gets \max\left\{0,  \text{CUSUM}_e[t-1] +  \log \left(\frac{{\bm f}_1^{(e)}({\bm X}_t)}{{\bm f}_0({\bm X}_t)}\right) \right\}$ 
    \EndFor
    \If{$\exists e \in E$ s.t. $\text{CUSUM}_{e}[t] \geq h$} \label{algo:fl-CUSUM:threshold}
        \State $\tau \gets t$, $\lambda \gets \arg\max_{e \in E} \text{CUSUM}_{e}[t]$
        \State Break
    \EndIf
\EndWhile
\State \Return $(\tau, \lambda)$
\end{algorithmic}
\end{algorithm}

\paragraph{Theoretical guarantees.}
To show the correctness of \FLCuSum, we verify that, by setting the threshold $h = \log((|E|+1)\gamma)$, the output of \FLCuSum, $(\tau, \lambda)$, satisfies the false alarm and localization requirement~\eqref{eq:WARL-Localization}.
In what follows, let $e^*$ denote the edge that suffers a transmission drop from time $\nu < \infty$.
First, we lower bound the average run length to false alarm by leveraging~\cite[Lemma 3]{mei2005information}, 
which states that, under the pre-change distribution, the probability that the CUSUM statistic $\text{CUSUM}_e > h$ is small for any $e \in E$, thereby ensuring a long average run length to false alarm.
\begin{lemma}[False alarm requirement]\label{thm:SR1-WARL}
$\mathbb{E}_\infty[\tau] \geq \gamma$.
\end{lemma}
\noindent
Although we are not yet able to formally verify the localization requirement in~\eqref{eq:WARL-Localization}, we believe that $\lambda$ indeed converges to $e^*$ using the same threshold $h$ in Lemma~\ref{thm:SR1-WARL}, as supported by the simulation results in \S~\ref{sec:simulation}, which we leave as a conjecture.
\begin{conjecture}[Localization accuracy]\label{thm:SR1-localization}
As $\gamma \rightarrow \infty$, $\lambda \rightarrow e^*$. 
\end{conjecture}

Finally, we derive an upper bound on the worst-case average detection delay.
Since for any $e \in E$, \FLCuSum's test statistics $\text{CUSUM}_e[t]$ is always non-negative, and $\text{CUSUM}_e[0] = 0$, the worst-case average detection delay happens when the change point $\nu = 1$.
By upper bounding $\mathbb{E}_1^{(e^*)}[\tau]$ using a generalized Weak Law of Large Numbers, we obtain the following guarantee.
\begin{theorem}[Detection delay]\label{thm:SR1-WADD}
As $\gamma \rightarrow \infty$, we have
    \begin{align*}
        \text{WADD}(\tau) \leq \frac{\log((|E|+1) \gamma)}{\sum_{P \in \mathcal{P}_{e^*}}D\left(f_{1,P}^{(e^*)}\big|\big|f_{0,P}\right)}(1+o(1)).
    \end{align*}
\end{theorem}
\noindent
We refer the reader to Appendices~\ref{sec:proof-WARL} and~\ref{sec:proof-WADD} for the proofs of Lemma~\ref{thm:SR1-WARL} and Theorem~\ref{thm:SR1-WADD}. 
In terms of the worst-case detection delay, the gap between the upper bound (Theorem~\ref{thm:SR1-WADD}) and the universal lower bound (Theorem~\ref{thm:universal}) is only an additive term of $\frac{\log(|E|+1)}{\sum_{P \in \mathcal{P}_{e^*}}D\left(f_{1,P}^{(e^*)}\big|\big|f_{0,P}\right)}(1+o(1))$.
It remains open whether this additive gap could be closed through a more refined analysis, thus establishing the asymptotic optimality of \FLCuSum.

\paragraph{Estimating the quantum speedup in reducing fault localization delay.}
Theorems~\ref{thm:universal} and~\ref{thm:SR1-WADD} suggest that, when we go from a channel to a network, rather than using the ratio $q_n$ (Lemma~\ref{lem:q_n}), we may estimate the quantum speedup in reducing fault localization delay using the ratio 
\begin{equation} \label{eqn:s_n}
s_n(e^*) = \frac{\sum_{P \in \mathcal{P}_{e^*}} D\left(\mathcal{N}_{n,P, e^*}^{(1)}||\mathcal{N}_{n,P}^{(0)}\right)}{n \cdot \sum_{P \in \mathcal{P}_{e^*}} D\left(\mathcal{N}_{P, e^*}^{(1)}||\mathcal{N}_P^{(0)}\right)},
\end{equation}
where the subscripts $P$ and $e^*$ are added to the notation of the Gaussian distributions in both the quantum~\eqref{eqn:quantum_Gaussian} and the classical~\eqref{eqn:classical_Gaussian} cases, to reflect the dependence of the transmissivity on specific choices of $P$ and $e^*$.
Although Theorem~\ref{thm:SR1-WADD} only provides an upper bound on the worst-case detection delay as the threshold $h$ approaches infinity, as we shall see in \S~\ref{sec:simulation}, $s_n(e^*)$ actually predicts the speedup factor quite accurately for finite and practical values of $h$.
\section{Simulation based evaluation} \label{sec:simulation}

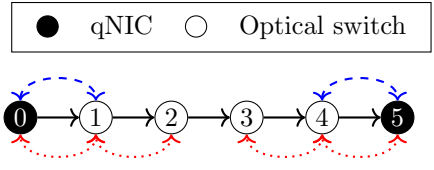
\begin{figure}
\centering
\resizebox{0.35\textwidth}{!}{%
\begin{tikzpicture}

\foreach \i in {0, 5} {
    \node[draw, circle, fill=black, inner sep=1.5pt] at (\i, 0) (v\i) {\textcolor{white}{\i}};
}
\foreach \i in {1, 2, 3, 4} {
    \node[draw, circle, fill=white, inner sep=1.5pt] at (\i, 0) (v\i) {\i};
}

\draw[line width=0.8pt, ->] (v0) -- (v1);  
\draw[line width=0.8pt, ->] (v1) -- (v2);  
\draw[line width=0.8pt, ->] (v2) -- (v3);  
\draw[line width=0.8pt, ->] (v3) -- (v4);  
\draw[line width=0.8pt, ->] (v4) -- (v5);  

\draw[<->, blue, thick, bend left=30, out=90, in=90, dashed] (v0) to (v1);  
\draw[<->, blue, thick, bend left=30, out=90, in=90, dashed] (v4) to (v5);

\draw[<->, red, thick, bend right=30, out=90, in=90, dotted] (v1) to (v0);
\draw[<->, red, thick, bend right=30, out=90, in=90, dotted] (v2) to (v1);
\draw[<->, red, thick, bend right=30, out=90, in=90, dotted] (v5) to (v4);
\draw[<->, red, thick, bend right=30, out=90, in=90, dotted] (v4) to (v3);

\node[draw, rectangle, fill=white, above right=20pt and -8pt of v0] (legend) {
    \begin{tabular}{llll}
        \raisebox{-1.5pt}{\tikz\draw[fill=black, draw=black] (0,0) circle(4pt);} & qNIC &
        \raisebox{-1.5pt}{\tikz\draw[fill=white, draw=black] (0,0) circle(4pt);} & Optical switch \\
    \end{tabular}
};
\end{tikzpicture}
}
\caption{5 probes to localize faults on a 5-link line network topology: $\textcolor{blue}{P_1}: 0\rightsquigarrow 1\rightsquigarrow 0$, $\textcolor{red}{P_2}: 0\rightsquigarrow 1\rightsquigarrow 2 \rightsquigarrow 1 \rightsquigarrow 0$, $\textcolor{blue}{P_3}: 5\rightsquigarrow 4\rightsquigarrow 5$, $\textcolor{red}{P_4}: 5\rightsquigarrow 4\rightsquigarrow 3\rightsquigarrow 4\rightsquigarrow 5$, and $P_5: 0\rightsquigarrow 1\rightsquigarrow 2\rightsquigarrow 3\rightsquigarrow 4\rightsquigarrow 5$.} 
\label{fig:line}
\end{figure}

\begin{figure*}[h]
\centering
\resizebox{\textwidth}{!}{%
\begin{tikzpicture}[
    tier0/.style={circle, fill=black, draw, inner sep=1.5pt, minimum size=15.2pt},
    tier1/.style={circle, fill=white, draw, inner sep=1.5pt},
    tier2/.style={circle, fill=white, draw, inner sep=1.5pt},
    tier3/.style={circle, fill=white, draw, inner sep=1.5pt},
    font=\small
]

\foreach \i in {0,...,15} {
  \node[tier0] (n\i) at (\i*1.1, 0) {\textcolor{white}{\i}};
}

\foreach \i [evaluate=\i as \x using (\i-16)*2.2+1] in {16,...,23} {
  \node[tier1] (n\i) at (\x, 1) {\i};
}

\foreach \i [evaluate=\i as \x using (\i-24)*2.2+1] in {24,...,31} {
  \node[tier2] (n\i) at (\x, 2) {\i};
}

\foreach \i [evaluate=\i as \x using (\i-32)*3.3+3] in {32,...,35} {
  \node[tier3] (n\i) at (\x, 3.2) {\i};
}

\foreach \i in {0,...,15} {
  \pgfmathtruncatemacro{\j}{16 + int(\i / 2)}
  \draw[line width=0.8pt, gray] (n\i) -- (n\j);
}

\foreach \i in {0,...,3} {
  \pgfmathtruncatemacro{\a}{16 + 2*\i}
  \pgfmathtruncatemacro{\b}{\a + 1}
  \pgfmathtruncatemacro{\sx}{24 + 2*\i}
  \pgfmathtruncatemacro{\sy}{\sx + 1}
  \foreach \agg in {\a, \b} {
    \foreach \spine in {\sx, \sy} {
      \draw[line width=0.8pt, gray] (n\agg) -- (n\spine);
    }
  }
}

\foreach \spine in {24,26,28,30} {
  \foreach \core in {32,33} {
    \draw[line width=0.8pt, gray] (n\spine) -- (n\core);
  }
}
\foreach \spine in {25,27,29,31} {
  \foreach \core in {34,35} {
    \draw[line width=0.8pt, gray] (n\spine) -- (n\core);
  }
}

\end{tikzpicture}
}
\caption{A $48$-link 3-tier fat-tree optical switched data center topology~\cite{al2008scalable}, where sending 48 loop-back probes originating and terminating at the qNIC layer (monitors) suffices to localize any single faulty link. For instance, we use length-$6$ loop-back probes:
$2i \leadsto (i+16) \leadsto (i+24) \leadsto 32$ and back,
as well as $(2i+1) \leadsto (i+16) \leadsto (i+25) \leadsto 34$ and back, for $i = 0,2,4,6$,
plus $2i \leadsto (i+16) \leadsto (i+23) \leadsto 33$ and back, 
as well as $(2i+1) \leadsto (i+16) \leadsto (i+24) \leadsto 35$ and back, for $i = 1,3,5,7$.
We also use the $32$ length-2 and length-4 probes nested in the above length-$6$ probes.
}
\label{fig:datacenter}
\end{figure*}
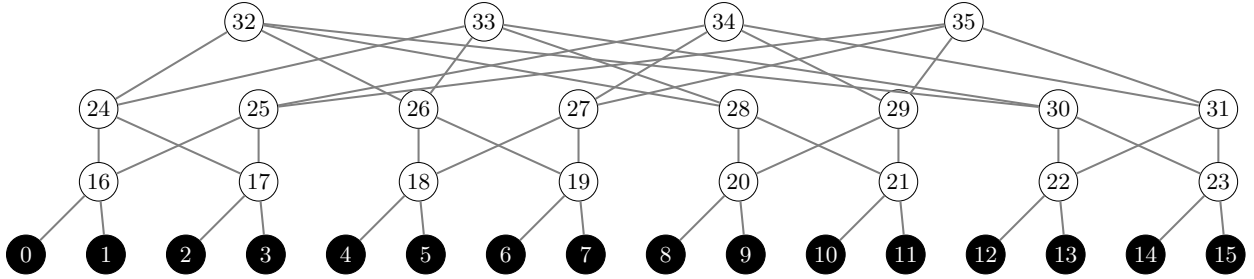

\begin{figure}
    \centering
    \begin{subfigure}[b]{0.49\textwidth}
        \centering
        \includegraphics[width=\linewidth]{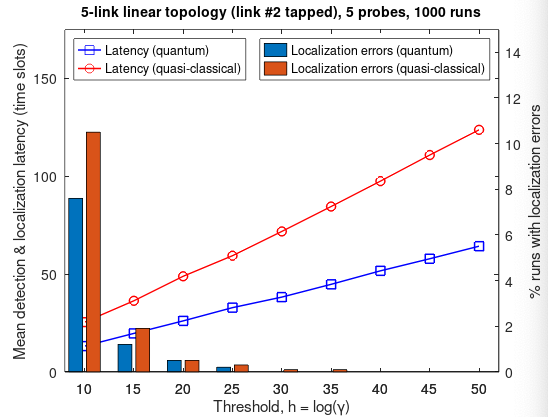}
        \caption{In a $5$-link line network topology of Figure~\ref{fig:line}, link $(1,2)$ suffers a transmission drop.}
        \label{fig:eval_line}
    \end{subfigure}
    \hfill
    \begin{subfigure}[b]{0.49\textwidth}
        \centering
        \includegraphics[width=\linewidth]{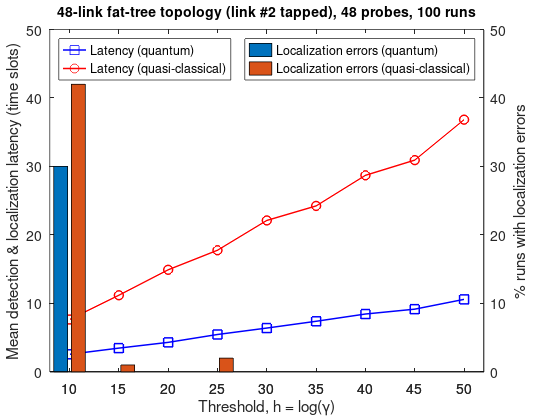}
        \caption{In a 3-tier fat-tree network topology of Figure \ref{fig:datacenter}, link $(1,16)$ suffers a transmission drop.}
        \label{fig:eval_datacenter}
    \end{subfigure}
    \caption{Detection/localization latency vs. localization errors (for $x$-labels, we have $h = \log((\sizeof{E}+1)\gamma)$).}
\end{figure}

We evaluate the performance of our network-level CUSUM algorithm first on the 5-link line network topology illustrated in Figure \ref{fig:line}. It is easy to see that 5 probes (one such scheme shown in Figure \ref{fig:line}) are necessary and sufficient to localize a single fault in this network. We perform an idealized simulation of this network where the transmitters of each probe use 6 dB of squeezing, each link is assumed to have pre-change transmissivity $\eta=0.9$ and a tap reduces the transmissivity of a probe by a factor of $\eta_{d}^m=0.95^m$, if the said probe traverses a tapped link $m\in \{1,2\}$ times.

The recipients of each probe share the homodyne measurements with a central controller which is responsible for network management. The controller runs \FLCuSum\, Algorithm~\ref{algo:fl-CUSUM}. Specifically, it tracks the CUSUM statistics from probe sets $\mathcal{P}_{(0,1)} = \{P_1,P_2,P_5\}$, $\mathcal{P}_{(1,2)} = \{P_2,P_5\}$, $\mathcal{P}_{(2,3)} = \{P_5\}$, $\mathcal{P}_{(3,4)} = \{P_4,P_5\}$, and $\mathcal{P}_{(4,5)} = \{P_3,P_4,P_5\}$ concurrently for localizing links $(0,1), (1,2), (2,3), (3,4)$ and $(4,5)$, respectively. Figure~\ref{fig:eval_line} shows the results of CUSUM from a scenario where link $(1,2)$ is faulty. We vary the cutoff threshold $h = \log((\sizeof{E}+1)\gamma)$ and show two metrics on the same plot: (1) mean latency (computed over 1000 runs) to detect a change and then localize it to a link (left-Y axis); and (2) probability of detection/localization error (right-Y axis). We record an error if we either declare any link other than link $(1,2)$ is faulty, or detect a change on link $(1,2)$ \textit{before} the ground truth change point (in our case it was $\nu=1000$). We observe high localization errors ($7-10\%$) for both quasi-classical and quantum (squeezed) probes at low values of $h$ (around 10), where the tolerance to false alarms is high. In this case, for the runs in which CUSUM successfully detects and localizes link $(1,2)$, the mean detection latency is low for both sets of probes. As the false alarm threshold is increased, the localization error drops drastically with the quantum probes slightly outperforming the quasi-classical ones. In contrast, mean  detection latency only increases linearly with $h$ for both probe sets as predicted by theory. This is natural since the CUSUM statistic needs more time to exceed a larger value of $h$. However, we observe that the slope of the quantum probes is about $2\times$ lower than that of the quasi-classical ones , thus pointing to a quantum speedup of $\approx 2\times$. An ideal scenario for the quantum probes is $h \sim 30$ where mean detection latency is low while the localization error is 0.

To show the relevance of our localization algorithm in a real world application, we illustrate an optical switched version of a 3-tier fat-tree network~\cite{al2008scalable} (see Figure \ref{fig:datacenter}) where the optical switches are MEMS devices made of movable tiny mirrors that allow one to control the light path between a pair of quantum network interfaces (qNIC) at any time instant. The MEMS assemblies can be reconfigured in tens of milliseconds, hence it is relatively easy to periodically reconfigure the light paths to align them with the intended probes. It is not hard to see that 48 loop-back probes provide exact \emph{identifiability} of all links in the network. We opt for the loop-back probes since a tapped link will see multiple probes traversing it twice, hence the difference between the post-change and pre-change distributions is more significant in this case than if a probe had traversed a tapped link only once.

The results of running concurrent CUSUMs is shown in Figure \ref{fig:eval_datacenter}. We observe that for $h=10$, both the quantum and classical probing schemes have over 30\%
localization errors. However errors drop drastically as the $h$ increases. For $h=50$, the squeezed probes detect the faulty link about 3.5 times faster than quasi-classical probes.
In both examples, our quantum speedup estimator $s_n(e^*)$ in~\eqref{eqn:s_n} matches the empirical speedup factors for finite threshold $h$.

Note that we do not have to send every single measurement from a probe to the central controller. 
To keep up with line rate, one may choose to sample the measurements.
A heavily sampled version of the measurements also contains the essential distributional information about various CUSUM statistics. Therefore one can trade off accuracy, latency, and network overhead. Studying such tradeoffs is part of our ongoing research.
\section{Related Work} \label{sec:related}

Our scheme builds on three components, the quantum probes (\S~\ref{sec:probes}), the probe construction algorithm (\S~\ref{sec:tomography}), and the \FLCuSum\, algorithm (\S~\ref{sec:qcd}).
The quantum probes are adopted from~\cite{guha2025quantum}.
In \S~\ref{sec:probes}, we have thoroughly discussed our theoretical results in relation to~\cite{guha2025quantum}.
In this section, we focus on previous works relevant to our path construction and \FLCuSum\, algorithms.

In Boolean network tomography, existing works have looked into the probe construction problem under different routing mechanisms.
In this work, we assume a probe may be any walk between monitors, which corresponds to the case of arbitrarily controllable routing (ACR)~\cite{cho2014localizing}.
Previous works have formulated the minimum-cost probe construction problem as an integer linear program (ILP)~\cite{ahuja2011srlg,stanic2010active,he2021network}.
While the ILP itself is general enough to accommodate different objectives (e.g., minimizing the total or average length of the probes, minimizing the total number of probes), the number of decision variables and the number of constraints can both be exponential in the size of the network, which makes solving the ILP intractable for practical networks.
Several works~\cite{ahuja2011srlg,stanic2010active} have proposed algorithms that compute feasible solutions of the ILP, but with no guarantee on how well the solutions approximate the optimum.
Though our algorithm shares the same structure as~\cite{ahuja2011srlg,stanic2010active}, our specific subroutines~\ref{algo:FindProbe} and~\ref{algo:FindOptProbe} allow us to solve Problem~\ref{prob} exactly.

Minimizing the number of probes has attracted broader interests, particularly in the context of non-adaptive group testing.
When all vertices in the graph are monitors, we enter the case of graph-constrained group testing~\cite{spang2018unconstraining,cheraghchi2012graph,harvey2007non}.
If we further assume the graph is complete, all claims in unconstrained group testing apply, and $\Theta(k^2 \log_k(\sizeof{E}))$ probes suffice to identify $k$ edge failures~\cite{d2014bounds}.
When probes are constrained to be walks on a graph, a construction via random walks~\cite{cheraghchi2012graph} shows that at most $O(c^4Tk^2\log{\frac{\sizeof{E}}{k}})$ probes are sufficient, for graphs with mixing time $T$ and some constant $c$, such that the degree $\deg(v)$ of each vertex $v \in V$ satisfies $6c^2 k T^2 \leq \deg(v) \leq 6c^3 k T^2$.
However, for a wide collection of graphs, $O(k^2\log{\frac{\sizeof{E}}{k}})$ probes suffice~\cite{spang2018unconstraining}.

When only a subset of the vertices are monitors, minimizing the number of probes is hard in the worst case~\cite[Corollary 6.21]{he2021network}.
We are not yet able to bound the number of probes the path construction algorithm (\S~\ref{sec:tomography:algo}) generates given a general graph $G$.
However, when $G$ is a path graph with $n$ vertices, and a monitor is placed at each endpoint, assuming there is at most one faulty edge in $G$, it is not hard to check that the algorithm outputs $n-1$ probes.
This matches the lower bound on the number of probes where at least half of the vertices are monitors, up to a factor of $2$~\cite[Theorem 1]{harvey2007non}.

In the QCD literature, the most relevant work to ours is~\cite{rovatsos2017statistical}, which addresses the problem of detecting and identifying line outages in a power network. While our detection procedure is similar to that of~\cite{rovatsos2017statistical} — both approaches maintain a CUSUM-based statistic for each possible attack or outage scenario and declare a change when any statistic exceeds a pre-specified threshold — our localization (or identification) procedure is notably different. Specifically, \cite{rovatsos2017statistical} produces a ranked list of candidate outaged lines upon detection, relying on subsequent inspection to determine the actual outage. 
The localization procedure in our scheme is more efficient. 
With a set of probes that ensure every faulty link is identifiable, our scheme directly finds the faulty link as the one associated with the first test statistic to cross the threshold.
\section{Conclusion} \label{sec:conclusion}

In this paper, we extend several previous results on squeezed-augmented probes to entanglement-augmented probes, and demonstrate an instance of quantum speedup within the context of networking.
Specifically, we propose a scheme that harnesses the power of such quantum probes for localizing transmission drops in optical networks, using quantum hardware that is readily available.
While we are able to construct a set of probes that can distinguish fault sets of arbitrary sizes, our quickest change localization algorithm currently limits the effectiveness of our scheme to a simplified scenario in which at most a single link may be faulty in a network.
However, it is our hope that, our approach---the combination of quantum probes, network tomography and the theory of QCD---could potentially lead to other instances of quantum speedup in optical network monitoring and beyond.
We leave that for future work.

\section*{Acknowledgments}

This research was supported by the DARPA Quantum Augmented Networking (QuANET) Program under contract number HR001124C0405 and by the DEVCOM Army Research Laboratory under Cooperative Agreement W911NF-17-2-0196, through the University of Illinois at Urbana-Champaign. The views, opinions and/or findings expressed are those of the authors and should not be interpreted as representing the official views or policies of the Department of Defense or the U.S. Government.

\newpage
\bibliographystyle{plain}
\bibliography{reference}

\newpage
\appendix
\section{The monotonicity and asymptotics of $q_n$ (Proof of Lemma~\ref{lem:q_n_monotonicity})}

For all claims on the monotonicity, we bound the corresponding partial derivatives.
\begin{enumerate}[label={(\arabic*)}]
\item 
Fix any $n$, $\eta_{d}$, $N$ and $N_a$.
Since $\frac{4Nn}{1-c_n \eta}$ increases in $\eta$, if we are able to show that
$f(\eta) 
= \frac{c_n (1-\eta_{d})}{1- c_n \eta}
-\frac{1}{\eta} \ln{\frac{1 - c_n\eta\eta_{d}}{1 - c_n\eta }}$
increases in $\eta$, 
we can conclude that $q_n$ increases in $\eta$.
We therefore look at $f'(\eta)$,
\begin{equation} \label{eqn:f(eta)_derivative}
f'(\eta)
= \frac{c_n^2(1-\eta_{d})}{(1-c_n\eta)^2}
- \frac{c_n(1-\eta_{d})}{\eta (1-c_n\eta)(1-c_n\eta \eta_{d})}
+ \frac{1}{\eta^2} \ln{\frac{1-c_n\eta \eta_{d}}{1-c_n\eta}}.
\end{equation}
Using $\ln{(y+1)} \geq \frac{y}{y+1}$ for all $y > -1$, we lower bound the third term in (\ref{eqn:f(eta)_derivative}),
\begin{equation} \label{lb:3rd_f(eta)_derivative}
\frac{1}{\eta^2} \ln{\frac{1-c_n\eta \eta_{d}}{1-c_n\eta}}
= \frac{1}{\eta^2} \ln{\left( 1+\frac{c_n\eta(1-\eta_{d})}{1-c_n\eta} \right)}
\geq \frac{c_n (1-\eta_{d})}{\eta (1-c_n\eta \eta_{d})}.
\end{equation}
By (\ref{eqn:f(eta)_derivative}), (\ref{lb:3rd_f(eta)_derivative}), and $c_n \in (0,1)$ for any $n$, $f'(\eta) \geq \frac{c_n^3 \eta (1-\eta_{d})}{(1-c_n \eta)^2 (1-c_n \eta \eta_{d})} > 0$.

When $\eta \to 0$, only the second term in (\ref{eqn:q_n}) (or rather, the second term in $f(n)$) requires applying L'Hôpital's rule,
\[
\lim_{\eta \to 0} -\frac{1}{\eta} \ln{\frac{1 - c_n\eta\eta_{d}}{1 - c_n\eta }}
= - \lim_{\eta \to 0} \frac{c_n (1-\eta_{d})}{(1-c_n\eta)(1-c_n\eta \eta_{d})}
= -c_n (1-\eta_{d}).
\]
Then the first and the second terms in $q_n$ magically cancel out at the limit, which gives (\ref{eqn:lim_eta_to_0}).

\item Fix any $n$, $\eta$, $N$ and $N_a$, and let
\[
g(\eta_{d}) 
= \frac{c_n}{1- c_n \eta} \cdot \frac{1-\eta_{d}}{(1-\sqrt{\eta_{d}})^2}
-\frac{1}{\eta} \cdot \frac{\ln{\frac{1 - c_n\eta\eta_{d}}{1 - c_n\eta }}}{(1-\sqrt{\eta_{d}})^2}.
\]
It suffices to show $g(\eta_{d})$ increases in $\eta_{d}$.
We again look at the derivative $g'(\eta_{d})$,
\begin{equation} \label{eqn:g(eta_tap)_derivative}
g'(\eta_{d})
= \frac{c_n (1- c_n \eta \sqrt{\eta_{d}}) (1-\eta_{d})-(1-c_n \eta)(1-c_n \eta \eta_{d})\ln{\frac{1-c_n \eta \eta_{d}}{1-c_n \eta}}}{\sqrt{\eta_{d}}(1-c_n\eta)(1-\sqrt{\eta_{d}})^{3}(1-c_n \eta \eta_{d})}.
\end{equation}
As the denominator of (\ref{eqn:g(eta_tap)_derivative}) is already positive, we focus on the numerator.
An appropriately chosen upper bound for $\ln{\frac{1-c_n \eta \eta_{d}}{1-c_n \eta}}$ allows us to lower bound the numerator by a function positive on $(0,1)$.
Here we use $\ln{(y+1)} \leq \frac{y}{2} \cdot \frac{y+2}{y+1}$, and get that
\begin{equation} \label{ineq:ub_ln}
0 < 
\ln{\frac{1-c_n \eta \eta_{d}}{1-c_n \eta}}
\leq \frac{c_n \eta (1-\eta_{d}) (2-c_n \eta(1+\eta_{d}))}{2(1-c_n \eta)(1-c_n \eta \eta_{d})}.
\end{equation}
Applying (\ref{ineq:ub_ln}) to the numerator of $g'(\eta_{d})$ (\ref{eqn:g(eta_tap)_derivative}), we have 
\begin{align*}
c_n (1- c_n \eta \sqrt{\eta_{d}}) (1-\eta_{d})
- & (1-c_n \eta)(1-c_n \eta \eta_{d})\ln{\frac{1-c_n \eta \eta_{d}}{1-c_n \eta}} \\
& \geq \frac{1}{2} c_n (1-\eta_{d}) (1 + c_n\eta (\sqrt{\eta_{d}} -1)^2) > 0.
\end{align*}

Next we show the limit is of the form (\ref{eqn:lim_eta_tap_to_0}).
By applying L'Hôpital's rule twice,
\begin{equation} \label{eqn:limit_g_eta_tap}
\lim_{\eta_{d} \to 1} g(\eta_{d}) 
= \frac{2c_n^2 \eta}{(1-c_n \eta)^2}.
\end{equation}
(\ref{eqn:lim_eta_tap_to_0}) follows from dressing (\ref{eqn:limit_g_eta_tap}) with constant factors in (\ref{eqn:q_n}).

\item Fix any $n$, $\eta$, $\eta_{d}$ and $N_a$, denote $c_1 = \frac{c_n}{1- c_n \eta} \cdot \frac{1-\eta_{d}}{(1-\sqrt{\eta_{d}})^2}
-\frac{1}{\eta} \cdot \frac{\ln{\frac{1 - c_n\eta\eta_{d}}{1 - c_n\eta }}}{(1-\sqrt{\eta_{d}})^2}$, $c_2 = \frac{4n}{1-c_n \eta}$, we have $q_n = \frac{1}{4n\left(N+N_{a}\right)}\left(c_{1}+c_{2}N\right)$, and $\frac{\partial}{\partial N} q_n = \frac{c_{2}a-c_{1}}{\left(N+N_{a}\right)^{2}}$.
Therefore, $q_n$ increases in $N$ when $c_{2}a-c_{1}>0$.
Again by (\ref{lb:3rd_f(eta)_derivative}), we can lower bound $c_{2}a-c_{1}$, and it is easy to check that, $N > \frac{c_n^2 t (1+\sqrt{\eta_{d}})^2}{4nN_a (1-c_n \eta \eta_{d})}$ makes the lower bound of $c_{2}a-c_{1}$ positive.

The limit (\ref{eqn:lim_qn_N}) follows trivially from the fact that $c_1, c_2$ are constants.

\item 
Fix any $\eta$, $\eta_{d}$, $N$ and $N_a < \frac{N\eta(1-\sqrt{\eta_{d}})}{(1-\eta)(1+\sqrt{\eta_{d}})}$.
Consider the derivative $\frac{\partial}{\partial n} q_n = \frac{1}{4(N+N_a)n^2}\sum_{i=0}^2 h_i(x)$, where
\begin{align*}
& h_0(n) = \frac{\sqrt{N_{a}n}\left(4\eta Nn-b\left(1+2N_{a}n+2\left(1-2\eta\right)\sqrt{N_{a}n+1}\sqrt{N_{a}n}\right)\right)}{\sqrt{N_{a}n+1}\left(\sqrt{N_{a}n+1}+\sqrt{N_{a}n}-2\eta\sqrt{N_{a}n}\right)^{2}}, \\
& h_1(n) = -\frac{b\sqrt{N_{a}n}}{\sqrt{N_{a}n+1}\left(\sqrt{N_{a}n+1}+\sqrt{N_{a}n}-2\eta\sqrt{N_{a}n}\right)\left(\sqrt{N_{a}n+1}+\sqrt{N_{a}n}-2\eta \eta_{d}\sqrt{N_{a}n}\right)}, \\
& h_2(n) = \frac{1}{\eta \left(1-\sqrt{\eta_{d}}\right)^{2}}\ln\frac{1-c_{n}\eta\eta_{d}}{1-c_{n}\eta}.
\end{align*}
Our goal is to show that, for $n$ that is not too small, both $h_0(n) > 0$ and $h_1(n) + h_2(n) > 0$.
While this is convenient and sufficient, it is not necessary.
Consequently, (\ref{eqn:n_0}) is not tight.
However, since $n$ only takes discrete values, we shall see that (\ref{eqn:n_0}) is already quite good for practical choices of parameters.

As in (\ref{lb:3rd_f(eta)_derivative}), we lower bound $h_2(n)$, and then $h_1(n)+h_2(n)$,
\[
h_2(n)
> \frac{2b\sqrt{N_{a}n}}{\sqrt{N_{a}n+1}+\sqrt{N_{a}n}-2\eta \eta_{d} \sqrt{N_{a}n}},
\]
\[
\resizebox{.87\textwidth}{!}{$
\begin{aligned}
h_1(n)+h_2(n)
& > \frac{b\sqrt{N_{a}n}\left(4\sqrt{N_{a}n+1}\sqrt{N_{a}n}\left(1-\eta\right)-1\right)}{\sqrt{N_{a}n+1}\left(\sqrt{N_{a}n+1}+\sqrt{N_{a}n}-2\eta\sqrt{N_{a}n}\right)\left(\sqrt{N_{a}n+1}+\sqrt{N_{a}n}-2\eta\eta_{d}\sqrt{N_{a}n}\right)} \\
& > \frac{b\sqrt{N_{a}n}\left(4N_a n\left(1-\eta\right)-1\right)}{\sqrt{N_{a}n+1}\left(\sqrt{N_{a}n+1}+\sqrt{N_{a}n}-2\eta\sqrt{N_{a}n}\right)\left(\sqrt{N_{a}n+1}+\sqrt{N_{a}n}-2\eta\eta_{d}\sqrt{N_{a}n}\right)}.
\end{aligned}
$}
\]
Therefore, $h_1(n)+h_2(n)>0$ if $4N_a n\left(1-\eta\right)-1 > 0$, which gives part of (\ref{eqn:n_0}).

To lower bound $h_0(n)$, we peel off positive factors in $h_0(n)$, and only consider 
\[
l_0(n) = 4\eta Nn-b\left(1+2N_{a}n+2\left(1-2\eta\right)\sqrt{N_{a}n+1}\sqrt{N_{a}n}\right).
\]
Depending on the choice of $\eta$, the lower bound of $l_0(n)$ slightly differs. 
For $\eta \geq \frac{1}{2}$,
\begin{align}
l_0(n) 
& > 4 N\eta n-b\left(1+2N_{a}n-2\left(2\eta-1\right)N_a n\right) \nonumber \\
& = 4(N\eta - b N_a(1-\eta))n - b. \label{ineq:lb:l0_eta_geq_1/2}
\end{align}
For $\eta < \frac{1}{2}$,
\begin{align}
l_0(n) 
& > 4 N\eta n-b\left(1+2N_{a}n+2\left(1-2\eta\right)(N_a n +1)\right) \nonumber \\
& = 4(N\eta - b N_a(1-\eta))n - b(3-4\eta). \label{ineq:lb:l0_eta_leq_1/2}
\end{align}
Given $N_a < \frac{N\eta(1-\sqrt{\eta_{d}})}{(1-\eta)(1+\sqrt{\eta_{d}})} = \frac{N\eta}{b(1-\eta)}$, $N\eta - b N_a(1-\eta)>0$, setting (\ref{ineq:lb:l0_eta_geq_1/2}) and (\ref{ineq:lb:l0_eta_leq_1/2}) to be positive gives the rest of (\ref{eqn:n_0}).

The limit (\ref{eqn:lim_n_q_n}) immediately follows from $\lim_{n \to \infty} c_n = 1$.

\item Fix any $n$, $\eta$, $\eta_{d}$ and $N$:
    \begin{enumerate}[label=(5.\arabic*)]
    \item
    For $n=1$, it is observed that the third term in $\frac{1}{n}D\left(\mathcal{N}_n^{(1)}||\mathcal{N}_n^{(0)}\right)$ (\ref{eqn:KL_quantum_n_final}) dominates if $N$ is large (Appendix A.3 in~\cite{guha2025quantum}).
    In finding an approximate range of $N_a$ where squeezing augmentation outperforms the classical benchmark, \cite{guha2025quantum} sets $q_1 \approx \frac{N}{(N+N_a)(1-c_1 \eta)} \approx \frac{N}{(N+N_a)(1-\eta)} >1$.
    Along similar lines but with rigor, here we show that for any $n$, the sum of the first two terms in $q_n$ (Lemma~\ref{lem:q_n}) is strictly positive:
    \begin{align*}
    & \frac{c_n}{1- c_n \eta} \cdot \frac{1-\eta_{d}}{(1-\sqrt{\eta_{d}})^2}
    -\frac{1}{\eta} \cdot \frac{\ln{\frac{1 - c_n\eta\eta_{d}}{1 - c_n\eta }}}{(1-\sqrt{\eta_{d}})^2} \\
    \geq & \frac{c_n}{1- c_n \eta} \cdot \frac{1-\eta_{d}}{(1-\sqrt{\eta_{d}})^2}
    -\frac{1}{\eta(1-\sqrt{\eta_{d}})^2} \cdot \frac{c_n \eta (1-\eta_{d}) (2-c_n \eta(1+\eta_{d}))}{2(1-c_n \eta)(1-c_n \eta \eta_{d})}
    && (\text{by (\ref{ineq:ub_ln})}) \\
    = & \frac{c_n^2\eta(1-\eta_{d})^2}{2(1-\sqrt{\eta_{d}})^2(1- c_n \eta)(1-c_n \eta \eta_{d})} > 0.
    \end{align*}
    This allows us to lower bound $q_n$ by its third term,
    \begin{equation} \label{ineq:lb_qn}
    q_n
    > \frac{N}{(N+N_a)(1-c_n \eta)}.
    \end{equation}
    For any $\varepsilon>0$, $N_a > \varepsilon$ implies that $c_n > c_{n,\varepsilon}$.
    Given $N_a < \frac{c_{n,\varepsilon}N\eta}{1-c_{n,\varepsilon}\eta}$, by (\ref{ineq:lb_qn}),
    \[
    q_n 
    > \frac{N}{(N+\frac{c_{n,\varepsilon}N\eta}{1-c_{n,\varepsilon}\eta})(1-c_{n,\varepsilon} \eta)}
    =1.
    \]
    
    \item 
    Once again, we consider the derivative $\frac{\partial}{\partial N_a} q_n = \frac{1}{4(N+N_a)^2}\sum_{i=0}^2 f_i(N_a)$, where
    \begin{align}
    & f_0(N_a) = 
    \frac{b_d\left(N+2N_{a}\left(2\eta-1\right)\sqrt{N_{a}n+1}\sqrt{N_{a}n}-2nN_{a}^{2}-N_{a}\right)}{\sqrt{N_{a}n+1}\sqrt{N_{a}n}\left(\sqrt{N_{a}n+1}+\sqrt{N_{a}n}-2\eta\sqrt{N_{a}n}\right)^{2}}, \label{eqn:f0(N_a)}\\
    & \resizebox{.84\textwidth}{!}{$
    \begin{aligned}
    f_1(N_a) = 
    & \frac{1}{\left(1-\sqrt{\eta_{d}}\right)^{2}n} \left( \frac{1}{\eta}\ln\left(\frac{1-c_{n}\eta \eta_{d}}{1-c_{n}\eta}\right) \right. \\
    &\qquad \left. -\frac{\left(1-\eta_{d}\right)n\left(N+N_{a}\right)}{\sqrt{N_{a}n+1}\sqrt{N_{a}n}\left(\sqrt{N_{a}n+1}+\sqrt{N_{a}n}-2\eta\sqrt{N_{a}n}\right)\left(\sqrt{N_{a}n+1}+\sqrt{N_{a}n}-2\eta\eta_{d}\sqrt{N_{a}n}\right)} \right),
    \end{aligned}
    $}
    \label{eqn:f1(N_a)} \\
    & \resizebox{.84\textwidth}{!}{$
    \begin{aligned}
    f_2(N_a) = 
    & -\frac{4N\left(\sqrt{N_{a}n+1}+\sqrt{N_{a}n}\right)}{\sqrt{N_{a}n+1}\sqrt{N_{a}n}\left(\sqrt{N_{a}n+1}+\sqrt{N_{a}n}-2\eta\sqrt{N_{a}n}\right)^{2}}\left( \sqrt{N_{a}n}+Nn\eta\left(\sqrt{N_{a}n}-\sqrt{N_{a}n+1}\right) \right. \\
    & \qquad \qquad \qquad \qquad \qquad \qquad \qquad \left. + N_{a}n\left(\sqrt{N_{a}n+1}+\sqrt{N_{a}n}+\eta\left(\sqrt{N_{a}n}-3\sqrt{N_{a}n+1}\right)\right) \right)
    .
    \end{aligned}
    $} \label{eqn:f2(N_a)}
    \end{align}
    To derive the limit, we look at the limit of each $f_i(N_a)$,
    \begin{align*}
    & \lim_{N_a \to 0} f_0(N_a)
    = b_dN \cdot \lim_{N_a \to 0} \frac{1}{\sqrt{N_a n}} 
    = -\lim_{N_a \to 0} f_1(N_a), \\
    & \lim_{N_a \to 0} f_2(N_a)
    = 4N^2 n \eta \cdot \lim_{N_a \to 0} \frac{1}{\sqrt{N_a n}}.
    \end{align*}
    Then the claim follows from $\lim_{N_a \to 0} \frac{\partial}{\partial N_a} q_n = \frac{1}{4(N+N_a)^2}\sum_{i=0}^2 \lim_{N_a \to 0} f_i(N_a)$.
    
    \item
    To show that $\frac{\partial}{\partial N_a} q_n < 0$ when $N_a$ is not too small, we upper bound each of the $f_i(N_a)$, $i=0,1,2$.
    As complicated as it seems, bounding $f_0(N_a)$ (\ref{eqn:f0(N_a)}) and $f_2(N_a)$ (\ref{eqn:f2(N_a)}) only requires applying $\sqrt{N_{a}n+1}\sqrt{N_{a}n} < N_{a}n+1$, and we get
    \begin{equation} \label{ineq:f0(N_a)_ub}
    f_0(N_a) < 
    \frac{b_d\left(-4n\left(1-\eta\right)N_{a}^{2}+\left(4\eta-1\right)N_{a}+N\right)}{\sqrt{N_{a}n+1}\sqrt{N_{a}n}\left(\sqrt{N_{a}n+1}+\sqrt{N_{a}n}-2\eta\sqrt{N_{a}n}\right)^{2}},
    \end{equation}
    \begin{equation} \label{ineq:f2(N_a)_ub}
    f_2(N_a) < 
    \frac{4N\left(-4n^{2}\left(1-\eta\right)a^{2}-\left(2-6\eta\right)na+Nn\eta\right)}{\sqrt{N_{a}n+1}\sqrt{N_{a}n}\left(\sqrt{N_{a}n+1}+\sqrt{N_{a}n}-2\eta\sqrt{N_{a}n}\right)^{2}}.
    \end{equation}
    For $f_1(N_a)$ (\ref{eqn:f1(N_a)}), we first upper bound the logarithmic term using (\ref{ineq:ub_ln}), 
    \begin{equation} \label{ineq:f1(N_a)_ub_temp}
    \resizebox{.87\textwidth}{!}{$
    \begin{aligned}
    f_1(N_a) < 
    \frac{2b_d\left(2-\eta\left(1+\eta_{d}\right)\right)nN_{a}^{2}+4N_{a}}{\sqrt{N_{a}n+1}\sqrt{N_{a}n}\left(\sqrt{N_{a}n+1}+\sqrt{N_{a}n}-2\eta\sqrt{N_{a}n}\right)\left(\sqrt{N_{a}n+1}+\sqrt{N_{a}n}-2\eta\eta_{d}\sqrt{N_{a}n}\right)}.
    \end{aligned}
    $}
    \end{equation}
    Noticing that the upper bound in (\ref{ineq:f1(N_a)_ub_temp}) has a positive numerator, we wiggle its denominator, so that all upper bounds of $f_i(N_a)$ have the same denominator,
    \begin{equation} \label{ineq:f1(N_a)_ub}
    f_1(N_a) < 
    \frac{2b_d\left(2-\eta\left(1+\eta_{d}\right)\right)nN_{a}^{2}+4N_{a}}{\sqrt{N_{a}n+1}\sqrt{N_{a}n}\left(\sqrt{N_{a}n+1}+\sqrt{N_{a}n}-2\eta\sqrt{N_{a}n}\right)^{2}}.
    \end{equation}
    By setting the sum of (\ref{ineq:f0(N_a)_ub}), (\ref{ineq:f2(N_a)_ub}) and (\ref{ineq:f1(N_a)_ub}) to be at most $0$, we obtain the upper bound on $N_a$ in (\ref{ineq:ub_Na}), when $8Nn\left(1-\eta\right)>\eta b_d\left(1-\eta_{d}\right)$.
    \end{enumerate}
\end{enumerate}
\section{Proofs Related to Quickest Change Detection and Localization}
\label{sec:proofs-qcd}

In this appendix, we fill in all the missing proofs from \S~\ref{sec:qcd}. 
Throughout this appendix, we assume edge $e^*$ becomes faulty from some timestep $\nu < \infty$.

\subsection{Universal detection delay lower bound (proof of Theorem~\ref{thm:universal})}
\label{sec:proofs-qcd:lb}

First recall the following generalized version of the Weak Law of Large Numbers.
\begin{lemma}[Lemma A.1 in~\cite{fellouris2017multichannel}]\label{lemma:WLLN}
    Let $\{Y_t, t\in \mathbb{N}\}$ be a sequence of random variables i.i.d. on $(\Omega, \mathcal{F}, \mathbb{P})$ with $\mathbb{E}[Y_t] = \mu > 0$, then for any $\epsilon > 0$, as $n\rightarrow \infty$, 
    \begin{align}
        \mathbb{P}\left[\frac{\max_{1\leq k \leq n} \sum_{t=1}^k Y_t}{n} - \mu > \epsilon\right] \rightarrow 0.
    \end{align}
\end{lemma}

In what follows, we lower bound the worst-case average detection delay. 
Note that
\begin{align}
    \text{WADD}(\tau) &:= \sup_{\nu \geq 1} \mathbb{E}_{\nu}^{(e^*)}[\tau -\nu|\tau \geq \nu]\notag\\
    &\geq \mathbb{E}_{\nu}^{(e^*)}[\tau - \nu|\tau \geq \nu]\notag\\
    &\overset{(a)}{\geq} \mathbb{P}_{\nu}^{(e^*)}[\tau - \nu \geq \alpha_\gamma|\tau \geq \nu] \times \alpha_\gamma, \label{eq:lower-begin}
\end{align}
where inequality (a) follows from the Markov inequality. It then suffices to show that as $\gamma \rightarrow \infty$,
\begin{align}
    \mathbb{P}_{\nu}^{(e^*)}[\tau - \nu \geq \alpha_\gamma|\tau \geq \nu] \rightarrow 1,
\end{align}
or equivalently, 
\begin{align}\label{eq:time-larger-prob}
    \mathbb{P}_{\nu}^{(e^*)}[ \nu \leq \tau < \nu + \alpha_\gamma|\tau \geq \nu] \rightarrow 0. 
\end{align}
    
We will first show that~\eqref{eq:time-larger-prob} holds when 
\begin{align}
    \alpha_\gamma = \frac{1}{D\left({\bm f}_{1}^{(e^*)}\big|\big|{\bm f}_{0}\right)+\epsilon}\log \gamma^{(1-\epsilon)}, \quad \epsilon > 0,
\end{align}
using a change-of-measure argument. 
Let $H_t := ({\bm X}_1,..., {\bm X}_t), t \in \mathbb{N}_+$. fFr some $a$ to be specified later, we have
\begin{align}
    \mathbb{P}_{0}[\nu \leq \tau < \nu + \alpha_\gamma] &= \mathbb{E}_{0}[\mathds{1}_{\{\nu \leq \tau < \nu +\alpha_\gamma\}}]\notag\\
    &\overset{(a)}{=} \mathbb{E}_{\nu}^{(e^*)}\left[\mathds{1}_{\{\nu \leq \tau < \nu +\alpha_\gamma\}}\frac{\mathbb{P}_0(H_\tau)}{\mathbb{P}_{\nu}^{(e^*)}(H_\tau)}\right]\notag\\
    &\geq \mathbb{E}_{\nu}^{(e^*)}\left[\mathds{1}_{\left\{\nu \leq \tau < \nu +\alpha_\gamma, \log \left(\frac{\mathbb{P}_0(H_\tau)}{\mathbb{P}_{\nu}^{(e^*)}(H_\tau)}\right)\geq -a\right\}}\frac{\mathbb{P}_0(H_\tau)}{\mathbb{P}_{\nu}^{(e^*)}(H_\tau)}\right]\notag\\
    &\geq e^{-a} \mathbb{P}_{\nu}^{(e^*)}\left[\nu \leq \tau < \nu +\alpha_\gamma, \log \left(\frac{\mathbb{P}_0(H_\tau)}{\mathbb{P}_{\nu}^{(e^*)}(H_\tau)}\right)\geq -a\right]\notag\\
    &= e^{-a} \mathbb{P}_{\nu}^{(e^*)}\left[\nu \leq \tau < \nu +\alpha_\gamma, \log \left(\frac{\mathbb{P}_{\nu}^{(e^*)}(H_\tau)}{\mathbb{P}_{0}(H_\tau)}\right)\leq a\right]\notag\\
    &\geq e^{-a} \mathbb{P}_{\nu}^{(e^*)}\left[\nu \leq \tau < \nu +\alpha_\gamma, \max\limits_{\nu \leq j < \nu +\alpha_\gamma}\log \left(\frac{\mathbb{P}_{\nu}^{(e^*)}(H_j)}{\mathbb{P}_{0}(H_j)}\right)\leq a\right]\notag\\
    & \overset{(b)}{\geq } e^{-a} \mathbb{P}_{\nu}^{(e^*)}[\nu \leq \tau < \nu +\alpha_\gamma] - e^{-a} \mathbb{P}_{\nu}^{(e^*)}\left[ \max\limits_{\nu \leq j < \nu +\alpha_\gamma}\log \left(\frac{\mathbb{P}_{\nu}^{(e^*)}(H_j)}{\mathbb{P}_{0}(H_j)}\right)> a\right], \label{eq:after-change-of-measure}
\end{align}
where the change-of-measure argument (a) holds since $\mathbb{P}_0$ and $\mathbb{P}_{\nu}^{(e^*)}$ are measures over a common measurable space, $\mathbb{P}_{\nu}^{(e^*)}$ is $\sigma$-finite, and $\mathbb{P}_0 \ll \mathbb{P}_{\nu}^{(e^*)}$; 
and inequality (b) is due to the fact that, for any event $A$ and $B$, $\mathbb{P}[A \cap B] \geq \mathbb{P}[A] - \mathbb{P}[B^c]$.

The event $\{\tau \geq \nu\}$ only depends on $H_{\nu-1}$, which follows the same distribution under $\mathbb{P}_0$ and $\mathbb{P}_{\nu}^{(e^*)}$. This implies 
\begin{align}
    \mathbb{P}_0[\tau \geq \nu] = \mathbb{P}_{\nu}^{(e^*)}[\tau \geq \nu]. \label{eq:same-prob-before-change}
\end{align}
By~\eqref{eq:after-change-of-measure} and reordering~\eqref{eq:same-prob-before-change}, it follows that 
\begin{equation} \label{eq:after-reordering}
\resizebox{.94\textwidth}{!}{$
\begin{aligned} 
\mathbb{P}_{\nu}^{(e^*)}[\nu \leq \tau < \nu + \alpha_\gamma | \tau \geq \nu] \leq e^{a} \mathbb{P}_0[\nu \leq \tau < \nu + \alpha_\gamma | \tau \geq \nu] + \mathbb{P}_{\nu}^{(e^*)}\left[\max\limits_{\nu \leq j < \nu + \alpha_\gamma}\log\left(\frac{\mathbb{P}_{\nu}^{(e^*)}(H_j)}{\mathbb{P}_0(H_j)}\right)> a\Big|\tau \geq \nu\right].
\end{aligned}
$}
\end{equation}

To show that the first term on the right-hand side of~\eqref{eq:after-reordering} converges to $0$ as $\gamma \rightarrow \infty$, we can use the proof-by-contradiction argument as in the proof of~\cite[Theorem 1]{L1998}.
Let $\alpha_\gamma$ be a positive integer and $\alpha_\gamma < \gamma$. For any $(\tau, \lambda) \in \mathcal{C}_\gamma$, we have $\mathbb{E}_0[\tau] \geq \gamma$.  Then for some $\nu \geq 1$, $\mathbb{P}_0[\tau \geq \nu] >0$, and
\begin{align}
    \mathbb{P}_0[\tau < \nu + \alpha_\gamma | \tau \geq \nu] \leq \frac{\alpha_\gamma}{\gamma}.
\end{align}
As otherwise, $\mathbb{P}_0[\tau \geq \nu + \alpha_\gamma] < 1 - \alpha_\gamma / \gamma$ for all $\nu \geq 1$ with $\mathbb{P}_0[\tau \geq \nu] >0$, which implies  that $\mathbb{E}_0[\tau] < \gamma$, contradicting the fact that $(\tau, \lambda) \in \mathcal{C}_\gamma$.
Let $a = \log \gamma^{1-\epsilon}$, then
\begin{align}
    e^a \mathbb{P}_0[\nu \leq \tau < \nu + \alpha_\gamma| \tau \geq \nu] \leq e^a \frac{\alpha_\gamma}{\gamma} = \frac{\alpha_\gamma}{\gamma^\epsilon} \rightarrow 0, \text{ as } \gamma \rightarrow \infty. 
\end{align}
    
We then show that the second term on the right-hand side of~\eqref{eq:after-reordering} also converges to $0$ as $\gamma \rightarrow \infty$.
\begin{align}
    &\mathbb{P}_{\nu}^{(e^*)}\left[\max\limits_{\nu\leq j < \nu + \alpha_\gamma}\log\left(\frac{\mathbb{P}_{\nu}^{(e^*)}(H_j)}{\mathbb{P}_0(H_j)}\right) > a \Big | \tau \geq \nu \right] \notag\\
    &= \mathbb{P}_{\nu}^{(e^*)}\left[\max\limits_{\nu\leq j < \nu + \alpha_\gamma}\sum_{i=\nu}^j\log\left(\frac{\mathbb{P}_{\nu}^{(e^*)}({\bm X}_i)}{\mathbb{P}_0({\bm X}_i)}\right) > a \Big | \tau \geq \nu \right]\notag\\
    & \overset{(a)}{=} \mathbb{P}_{\nu}^{(e^*)}\left[\max\limits_{\nu\leq j < \nu + \alpha_\gamma}\sum_{i=\nu}^j\log\left(\frac{\mathbb{P}_{\nu}^{(e^*)}({\bm X}_i)}{\mathbb{P}_0({\bm X}_i)}\right) > a\right]\notag\\ 
    &\overset{(b)}{\leq} \mathbb{P}_{\nu}^{(e^*)}\left[\max\limits_{\nu\leq j < \nu + \alpha_\gamma}\sum_{i=\nu}^j\log\left(\frac{\mathbb{P}_{\nu}^{(e^*)}({\bm X}_i)}{\mathbb{P}_0({\bm X}_i)}\right) > \alpha_\gamma \left(D\left({\bm f}_{1}^{(e^*)}\big|\big|{\bm f}_{0}\right)+\epsilon\right)\right] \rightarrow 0, \text{ as } \gamma \rightarrow \infty, \label{eq:lower-end}
    \end{align}
where equality (a) is due to the fact that the event $\{\tau \geq \nu\}$ is independent from ${\bm X}_i, \forall i \geq \nu$; inequality (b) is because $a \geq \alpha_\gamma \left(D\left({\bm f}_{1}^{(e^*)}\big|\big|{\bm f}_{0}\right)+\epsilon\right)$; and the last step is due to Lemma~\ref{lemma:WLLN}. 

By~\eqref{eq:lower-begin}-\eqref{eq:lower-end},
\begin{align}
    \text{WADD}(\tau) 
    \geq \frac{\log \gamma}{D\left({\bm f}_{1}^{(e^*)}\big|\big|{\bm f}_{0}\right) +o(1)}(1-o(1))
    & \overset{(a)}{=} \frac{\log \gamma}{\sum_{p \in \mathcal{P}_{e^*}}D\left(f_{1,p}^{(e^*)}\big|\big|f_{0,p}\right)+o(1)}(1-o(1)) \\
    & = \frac{\log \gamma}{\sum_{P \in \mathcal{P}_{e^*}}D\left(f_{1,P}^{(e^*)}\big|\big|f_{0,P}\right)}(1+o(1)),
\end{align}
where (a) follows from the divergence decomposition lemma~\cite[Lemma 15.1]{lattimore2020bandit}. This completes the proof.

\subsection{Run length to false alarm lower bound (proof of Lemma~\ref{thm:SR1-WARL})}\label{sec:proof-WARL}
    
We lower bound the average run length to false alarm of \FLCuSum\, as follows:
\begin{align*}
\mathbb{E}_\infty[\tau] 
&= \sum_{t=0}^\infty \mathbb{P}_\infty[\tau \geq t]\\
&=\sum_{t=0}^\infty \prod_{e \in E}\mathbb{P}_\infty[\tau^{\text{CUSUM}_e} \geq t]
&& (\text{by Step~\ref{algo:fl-CUSUM:threshold} of \FLCuSum})\\
&\geq \sum_{t=0}^\infty \left((1-t\cdot\exp(-h))^+\right)^{|E|}
&& (\text{by~\cite[Lemma 3]{mei2005information}})\\
&= \sum_{t=0}^{\exp(h)} (1-t\cdot\exp(-h))^{|E|}\\
&\geq \int_{0}^{\exp(h)} (1-t\cdot\exp(-h))^{|E|} dt
&& (\text{by the left Riemann sum})\\
&= \frac{{\exp(h)}}{|E|+1}\\
&= \gamma,
\end{align*}
where the last step follows from setting the threshold $h = \log((|E|+1)\gamma)$.

\subsection{Detection delay upper bound (proof of Theorem~\ref{thm:SR1-WADD})}\label{sec:proof-WADD}

Here we upper bound the worst-case average detection delay of \FLCuSum. 
As we have seen in \S~\ref{sec:qcd:algo}, the worst-case average detection delay happens when the change-point $\nu = 1$.
Let the stopping time of the test statistic $\text{CUSUM}_e$ for link $e \in E$ be
\begin{align}
    \tau^{\text{CUSUM}_e} := \inf\{t\geq 1: \text{CUSUM}_e[t]\geq h\}.
\end{align}
Then, by~\eqref{eq:flcusum-stopping-time}, the stopping time of \FLCuSum\, is given by
\begin{align}\label{eq:stopping-relation}
    \tau = \min\limits_{e \in E} \left\{ \inf\{ t \geq 1 : \text{CUSUM}_e[t] \geq b \} \right\}= \min\limits_{e \in E} \tau^{\text{CUSUM}_e}.
\end{align}
Hence, 
\begin{align}
    \mathbb{E}_{1}^{(e^*)}[\tau] \leq \min_{e \in E}\mathbb{E}_{1}^{(e^*)}[\tau^{\text{CUSUM}_e}]\leq \mathbb{E}_1^{(e^*)}[\tau^{\text{CUSUM}_{e^*}}]. \label{eq:wadd-upper-middle}
\end{align}
Next, we upper bound the worst-case average detection delay of $\text{CUSUM}_{e^*}$. Specifically, for any $e \in E$, there is some $\alpha_e \geq 0$ to be specified later, such that we have
\begin{align}
\mathbb{E}_{1}^{(e^*)}\left[\nicefrac{\tau^{\text{CUSUM}_{e}}}{\alpha_e}\right] 
&= \sum_{i=0}^\infty \mathbb{P}_{1}^{(e^*)}\left[\nicefrac{\tau^{\text{CUSUM}_e}}{\alpha_e} > i\right] \notag\\
&\leq 1 + \sum_{i=1}^\infty \mathbb{P}_{1}^{(e^*)}\left[\nicefrac{\tau^{\text{CUSUM}_e}}{\alpha_e} > i\right] \notag\\
&= 1 + \sum_{i=1}^\infty \mathbb{P}_{1}^{(e^*)}\left[\forall 1\leq t \leq i\alpha_e: \text{CUSUM}_e[t] < h\right] \notag\\
&\leq 1 + \sum_{i=1}^\infty \mathbb{P}_{1}^{(e^*)}\left[\bigcap_{j=1}^i (\text{CUSUM}_e[j \alpha_e] < h)\right] \notag\\
&\overset{(a)}{\leq } 1 + \sum_{i=1}^\infty\prod_{j=1}^i \mathbb{P}_{1}^{(e^*)}\left[\max_{k:(j-1)\alpha_e + 1\leq k \leq j \alpha_e} \sum_{t=k}^{j\alpha_e} \log\left(\frac{{\bm f}_{1}^{(l)}({\bm X}_t)}{{\bm f}_0({\bm X}_t)}\right) < h \right], \label{eq:dd-length}
\end{align}
where ${\bm X}_t \sim {\bm f}_{1}^{(e^*)}$, and inequality (a) follows from the definition of $\text{CUSUM}_e[t]$, and the fact that the random variables are independent.
For $e = e^*$, let
\begin{align}
    \alpha_{e^*} = \frac{h}{D\left({\bm f}_{1}^{(e^*)}\big|\big|{\bm f}_0\right)}(1+\epsilon), \quad \epsilon > 0.
\end{align}
Then, it follows from Lemma~\ref{lemma:WLLN} that, as $h \rightarrow \infty$,
\begin{align}
\frac{\max\limits_{k:1\leq k \leq \alpha_{e^*}}\sum_{t = k}^{\alpha_{e^*}} \log\left(\frac{{\bm f}_{1}^{(e^*)}({\bm X}_t)}{{\bm f}_0({\bm X}_t)}\right)}{h}\overset{p}{\rightarrow}\beta,
\end{align}
where ${\bm X}_t \sim {\bm f}_{1}^{(e^*)}$, and $\beta >  1$. 
Therefore, as $h \rightarrow \infty$,
\begin{align}
\mathbb{P}_{1}^{(e^*)}\left[\max\limits_{k:1\leq k \leq \alpha_{e^*}}\sum_{t = k}^{\alpha_{e^*}} \log\left(\frac{{\bm f}_{1}^{(e^*)}({\bm X}_t)}{{\bm f}_0({\bm X}_t)}\right)< h \right]\rightarrow 0.
\end{align}
This implies that
\begin{align}
\mathbb{P}_{1}^{(e^*)}\left[\max\limits_{k:1\leq k \leq \alpha_{e^*}} \sum_{t = k}^{\alpha_{e^*}} \log\left( \frac{{\bm f}_{1}^{(e^*)}({\bm X}_t)}{{\bm f}_0({\bm X}_t)}\right) < h\right] \leq \delta, \label{eq:prob-nopassing}
\end{align}
where $\delta$ can be arbitrarily small for large $h$.
By~\eqref{eq:prob-nopassing} and~\eqref{eq:dd-length}, we have
\begin{align}                       
    \mathbb{E}_{1}^{(e^*)}\left[\nicefrac{\tau^{\text{CUSUM}_{e^*}}}{\alpha_{e^*}}\right] &\leq 1 + \sum_{i=1}^\infty \delta^{i} = \frac{1}{1-\delta}, 
\end{align}
which implies that, as $h \rightarrow \infty$,
\begin{align}\label{eq:attacked-DD}
\mathbb{E}_{1}^{(e^*)}\left[\tau^{\text{CUSUM}_{e^*}}\right] \leq \frac{\alpha_{e^*}}{1-\delta} \leq \frac{h}{D\left({\bm f}_{1}^{(e^*)}\big|\big|{\bm f}_0\right)}(1+o(1)) \overset{(b)}{=} \frac{h}{\sum_{p \in \mathcal{P}_{e^*}} D\left(f_{1,p}^{(e^*)}\big|\big|f_{0,p}\right)}(1+o(1)),
\end{align}
where equality (b) follows from the divergence decomposition lemma~\cite[Lemma 15.1]{lattimore2020bandit}.

Setting $h = \log((|E|+1)\gamma)$, by~\eqref{eq:wadd-upper-middle} and~\eqref{eq:attacked-DD}, we have, as $\gamma \rightarrow \infty$,
\begin{align*}
\mathbb{E}_{1}^{(e^*)}[\tau] \leq \min_{e \in E}\mathbb{E}_{1}^{(e^*)}[\tau^{\text{CUSUM}_e}]\leq \mathbb{E}_1^{(e^*)}[\tau^{\text{CUSUM}_{e^*}}]\leq \frac{\log((|E|+1)\gamma)}{\sum_{P \in P_{e^*}} D(P_{1,p}||P_{0,p})}(1+o(1)). 
\end{align*}
This completes the proof.

\end{document}